\documentclass[journal,12pt,onecolumn,letterpaper,draftclsnofoot, margin=1in]{IEEEtran}

\def\BibTeX{{\rm B\kern-.05em{\sc i\kern-.025em b}\kern-.08em
		T\kern-.1667em\lower.7ex\hbox{E}\kern-.125emX}}

\usepackage{microtype}
\usepackage{graphicx}
\usepackage{subcaption}
\usepackage{booktabs} 

\usepackage[utf8]{inputenc} 
\usepackage[T1]{fontenc}
\usepackage{url}
\usepackage{ifthen}
\usepackage[noadjust]{cite}
\usepackage[cmex10]{amsmath} 
\usepackage{amsthm,amssymb,amsfonts}
\usepackage{algorithm}
\usepackage[noend]{algpseudocode}
\usepackage{graphicx}
\usepackage{textcomp}
\usepackage{xcolor}
\usepackage{comment}
\usepackage{epstopdf}
\usepackage{float}
\usepackage{amssymb}
\usepackage{color}
\usepackage{relsize}
\usepackage{enumitem}
\usepackage{cleveref}
\usepackage{cases}
\usepackage[nocomma]{optidef}
\usepackage[mode=buildnew]{standalone}
\usepackage{breqn}
\usepackage{diagbox}

\usepackage{tikz}
\usetikzlibrary{calc}

\usepackage{caption}
\setlist{leftmargin=4.1mm}

\newcommand{\tikzmark}[1]{\tikz[overlay,remember picture] \node (#1) {};}
\newcommand{\DrawBox}[3][]{%
	\tikz[overlay,remember picture]{
		\draw[black,#1]
		($(#2)+(-0.5em,2.0ex)$) rectangle
		($(#3)+(0.75em,-0.75ex)$);}
}

\colorlet{group1}{orange}
\colorlet{group2}{blue}
\colorlet{group3}{red}
\colorlet{group4}{violet}

\theoremstyle{plain}
\newtheorem{theorem}{Theorem}
\newtheorem{lemma}[theorem]{Lemma}
\newtheorem{definition}[theorem]{Definition}
\newtheorem{remark}{Remark}
\newtheorem{corollary}[theorem]{Corollary}

\theoremstyle{definition}

\newcommand{\widebf}[1]{\widetilde{\mathbf{#1}}}

\newcommand\lev[1]{{\color{black}#1}}

\begin{document}
	
	\title{Variable Coded Batch Matrix Multiplication\vspace{-0.5cm}
	}
	
	\author{\IEEEauthorblockN{Lev Tauz and Lara Dolecek}\\
		\IEEEauthorblockA{Department of Electrical and Computer Engineering, \\
			University of California, Los Angeles \\ 
			levtauz@ucla.edu, dolecek@ee.ucla.edu 
			\vspace{-0.5cm}
		}
	}
	
	\maketitle

	\begin{abstract}
	 A majority of coded matrix-matrix computation literature has broadly focused in two directions: matrix partitioning for computing a single computation task and batch processing of multiple distinct computation tasks. While these works provide codes with good straggler resilience and fast decoding for their problem spaces, these codes would not be able to take advantage of the natural redundancy of re-using matrices across batch jobs. In this paper, we introduce the Variable Coded Distributed Batch Matrix Multiplication (VCDBMM) problem which tasks a distributed system to perform batch matrix multiplication where matrices are not necessarily distinct among batch jobs. Inspired in part by Cross-Subspace Alignment codes, we develop Flexible Cross-Subspace Alignments (FCSA) codes that are flexible enough to utilize this redundancy. We provide a full characterization of FCSA codes which allow for a wide variety of system complexities including good straggler resilience and fast decoding. We theoretically demonstrate that, under certain practical conditions, FCSA codes are within a factor of 2 of the optimal solution when it comes to straggler resilience. Furthermore, our simulations demonstrate that our codes can achieve even better optimality gaps in practice, even going as low as 1.7. 
	\end{abstract}
	
	\vspace{-0.1cm}
	\section{Introduction and Motivation}
	\vspace{-0.1cm}
	
	Large scale distributed matrix-matrix multiplication is a fundamental component of modern data analytics and is used to deal with the exponential rise of big data. Yet, the presence of \textit{stragglers}  (i.e., workers that fail or are slow to respond) significantly hampers distributed systems due to the increase in tail latency\cite{dean2013Tail}. To mitigate stragglers, researchers in the field of coded computation inject computational redundancy through the use of error-correcting codes and have developed a plethora of various coded computation strategies for matrix-matrix computations
	\cite{Yu2017PolynomialCA,Dutta2020OnTO,yu2020straggler,Dutta2018AUC,Wang2018CodedSM,Sheth2018AnAO,Jeong2018LocallyRC,Baharav2018StragglerProofingMD,Yu2020EntangledPC,Hasircioglu2020BivariatePC,Malladi2019ACV,Subramaniam2019RandomKC,Das2019RandomCC,jia2021cross,Yu2019LagrangeCC,fahim2021numerically,sun2021coded,ramamoorthy2021numerically,wang2021price,cadambe2021approximate}, for a survey of these exciting results please refer to \cite{RamamoorthyDT20}. 
	Codes for matrix-matrix computations are broadly separated into two problem spaces: i) matrix partitioning for computing a single computation task \cite{Yu2017PolynomialCA,Dutta2020OnTO,yu2020straggler}, and ii) batch processing of multiple distinct computation tasks \cite{jia2021cross,Yu2019LagrangeCC,Yu2020EntangledPC}. We can summarize one example of problem space (i) as using coding to determine the product $\mathbf{A}\mathbf{B}$ by coding across the row partition $\mathbf{A}^T = (\mathbf{A}^T_i)^n_{i=1}$ and the column partition $\mathbf{B} = (\mathbf{B}_i)^m_{i=1}$ to determine $(\mathbf{A}_i\mathbf{B}_j)^{n,m}_{i=1,j=1}$ \cite{Yu2017PolynomialCA} (note that the partitions are matrices). Similarly, we summarize problem space (ii) by a  system that receives two lists of matrices $(\mathbf{A}_i)^p_{i=1}, (\mathbf{B}_i)^p_{i=1}$ and the goal is to use coding to determine $(\mathbf{A}_i\mathbf{B}_i)^{p}_{i=1}$ \cite{Yu2019LagrangeCC,jia2021cross}.  While state-of-the-art codes for these problem spaces provide near optimal straggler resilience, ultimately they rely on the rigid structure of the computation tasks they aim to compute and are hard to extrapolate to more variable tasks. For example, a variable computation task could require the following matrix products $(\mathbf{A}_1\mathbf{B}_1,\mathbf{A}_1\mathbf{B}_2,\mathbf{A}_2\mathbf{B}_2,\mathbf{A}_2\mathbf{B}_3)$ which is clearly not well suited to the previous two problem spaces.  We seek to generalize these two problem spaces of matrix-matrix computation and provide a novel coding scheme. 
	
	In this work, we introduce the Variable Coded Distributed Batch Matrix Multiplication (VCDBMM) problem that generalizes the two stated problem spaces. Assume that a distributed system is provided with two sets of matrices $\mathcal{A}=\{\mathbf{A}_1,\mathbf{A}_2,\dots,\mathbf{A}_{|\mathcal{A}|}\}$ and  $\mathcal{B}=\{\mathbf{B}_1,\mathbf{B}_2,\dots,\mathbf{B}_{|\mathcal{B}|}\}$ and a set of computation goals $\mathcal{S} = \{(i_1,j_1),(i_2,j_2), \dots (i_{|\mathcal{S}|},j_{|\mathcal{S}|})\}$ where the objective is to calculate the matrix multiplication $\mathbf{A}_i\mathbf{B}_j$ for every $(i,j) \in \mathcal{S}$. This model is a natural expansion of the previous models and can be included in a broad range of applications where data re-usability is common such as recommender systems \cite{Smith_Amazon_2017} and multi-model training \cite{ketkar2017stochastic}. %
	1) \textbf{Recommender Systems with Linear Classifiers} \cite{zhang2002recommender,Smith_Amazon_2017} -  Matrix $\mathbf{A}_i$ ($\mathbf{B}_j$) represents the $i^{\text{th}}$ user ($j^{\text{th}}$ linear recommender algorithm) and the system wishes to determine different recommendations for different users based on the user's characteristics. For example, a consumer website, such as Amazon, would want to determine different recommendations of products for college students in comparison to senior citizens. 
	2) \textbf{Concurrent Machine Learning Model Training with Stochastic Gradient Descent (SGD)} \cite{ketkar2017stochastic}- SGD is a variant of gradient descent where the gradient step at each iteration only depends on a small subset of the dataset. Recent research has determined that sampling without replacement can improve the performance of SGD \cite{safran2020good,rajput2021permutation}, and determining the batch order is now an important hyperparameter. Concurrent training of multiple AI models, such as linear regression, with different SGD orderings fits within the space of VCDBMM since each model would apply computations on a different subset of the data.  Essentially, VCDBMM applies to problems where there is a lot of data re-usability in concurrent matrix product operations.

	To solve the VCDBMM problem, we introduce Flexible Cross-Subspace Alignment (FCSA) codes (inspired in part by Cross-Subspace Alignment codes \cite{jia2021cross}) that take advantage of the redundancy in the VCDBMM problem to provide straggler resilience. The main idea behind FCSA codes is partitioning the desired computations into groups and carefully coding these groups using rational functions to limit the interference caused by undesired computations. We provide a full characterization of FCSA codes in terms of key system parameters such as communication cost, worker computation complexity, and decoding complexity. A key parameter that we focus on is the recovery threshold which is the minimum number of worker outputs needed at the fusion node so that it may recover the desired computations. A small recovery threshold results in higher straggler resilience. Due to the variable nature of the VCDBMM problem, FCSA codes do not provide a single solution but a space of coding solutions that have to be optimized to determine the best recovery threshold. We provide a methodology and simple constructions within the FCSA coding space that can be applied to any VCDBMM problem. While these constructions may not necessarily provide the optimal recovery threshold, we demonstrate that, under practically relevant conditions, FCSA codes are near-optimal and achieve a recovery threshold within a factor of $2$ of the optimal solution under any realization of the VCDBMM problem and outperform naive applications of existing methods. Furthermore, we simulate instances of the VCDBMM problem and demonstrate that in practice FCSA codes achieve even better factors of optimality, even as low as $1.7$. 
	
	This paper is organized as follows. In Section \ref{sec:prelim}, we define the system model and provide background on relevant constructions. In Section  \ref{sec:fcsa}, we construct our base FCSA codes and provide illustrative examples that demonstrate both how to explicitly construct the codes and the benefits over the existing schemes in terms of the recovery threshold. In Section \ref{sec:fcsa_flex}, we extend our base FCSA codes to provide more flexibility in communication and computation complexity at the cost of a higher recovery threshold. In Section \ref{sec:lower_bound}, we prove a lower bound on the minimum possible recovery threshold for VCDBMM.  In Section \ref{section:fcsa_types}, we provide special cases of FCSA codes that are provably within a multiplicative factor of $2$ with respect to the mentioned lower bound which is comparable to existing schemes. In Section \ref{sec:sim}, we provide numerical simulations that demonstrate that our FCSA codes can in practice achieve optimality gaps that are much smaller than the theoretical bound implies, and, therefore, can outperform previous state of the art codes in terms of the recovery threshold.
	
	\textit{Notation:} Let boldface capital letters represent matrices. Let $\mathbb{Z}_{>}$ and $\mathbb{Z}_{\geq}$ denote the set of positive and non-negative integers, respectively. 
	For any $n \in \mathbb{Z}_{>}$, $[n]$ denotes the set $\{1,2,\dots,n\}$. 
	Given two sets $A$ and $B$, $A \times B$ is the Cartesian product of the two sets. The notation $\widetilde{\mathcal{O}}(a\log^2b)$ suppresses poly log terms, i.e.,  $\widetilde{\mathcal{O}}(a\log^2b) = \mathcal{O}(a\log^2b\log\log(b))$.

	\vspace{-0.1cm}
	\section{System Model and Preliminaries}\label{sec:prelim}
	\vspace{-0.1cm}
	
	\begin{figure}[t]
		\centering
		\includegraphics[width= \linewidth]{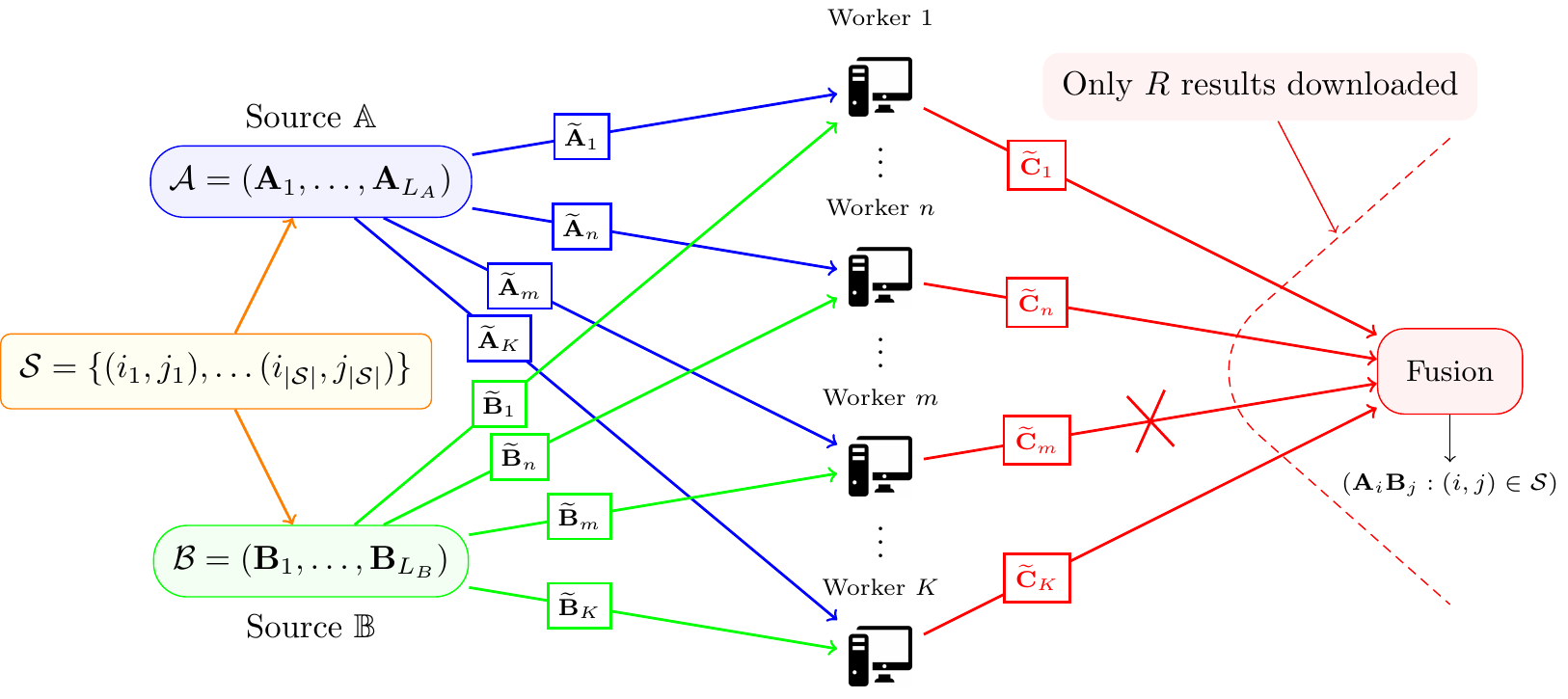}
		\vspace{-0.1cm}
		\caption{System Model for Variable Coded Distributed Batch Matrix Multiplication with a recovery threshold of $R$.}
		\label{fig:system_model}	
		\vspace{-0.5cm}
	\end{figure}

	We now define the Variable Coded Distributed Batch Matrix Multiplication (VCDBMM) problem. As shown in Fig.~\ref{fig:system_model}, consider two source nodes each of which generates a set of matrices $\mathcal{A}=\{\mathbf{A}_1,\dots,\mathbf{A}_{L_A}\}$ and $\mathcal{B}=\{\mathbf{B}_1,\mathbf{B}_2,\dots,\mathbf{B}_{L_B}\}$ such that for all $i \in [L_A]$ and $j\in [L_B]$ we have $\mathbf{A}_i\in \mathbb{F}^{\alpha \times \beta}$ and $\mathbf{B}_j\in \mathbb{F}^{\beta \times \gamma}$ where $\alpha,\beta,\gamma \in \mathbb{Z}_{>}$ and $\mathbb{F}$ is a finite field. We name these nodes as Source $\mathbb{A}$ and Source $\mathbb{B}$. Additionally, a computation list is provided  $\mathcal{S} = \{(i_1,j_1),(i_2,j_2), \dots (i_{|\mathcal{S}|},j_{|\mathcal{S}|})\}$  where the objective is for the fusion node to obtain the matrix multiplication $\mathbf{A}_i\mathbf{B}_j$ for every $(i,j) \in \mathcal{S}$. To accomplish this goal, the system has $K$ worker nodes which perform the bulk of the computation. To avoid degeneracy, we assume that $\beta \geq \max(L_A\alpha,L_B\gamma)$. Additionally, $\mathcal{S}$ must contain an index for every matrix in $\mathcal{A}$ and $\mathcal{B}$, otherwise we may simply prune this matrix without affecting the final computation.

	Let $\mathbf{f} = (f_1,f_2,\dots,f_K)$ be the set of encoding functions for Source $\mathbb{A}$ and let $\mathbf{g} = (g_1,g_2,\dots,g_K)$ be the set of encoding functions for Source $\mathbb{B}$. For the $k^{\text{th}}$, $1\leq k \leq K$, worker, the source nodes transmit $\widebf{A}_k = f_k(\mathcal{A}) , \widebf{B}_k = g_k(\mathcal{B}),$
	where $\widebf{A}_k \in \mathbb{F}^{\widetilde{\alpha}_k \times \widetilde{\beta}_k}$ and $\widebf{B}_k \in \mathbb{F}^{\widetilde{\beta}_k \times \widetilde{\gamma}_k}$ for some $\widetilde{\alpha}_k,\widetilde{\beta}_k,\widetilde{\gamma}_k \in \mathbb{Z}_{>}$. We assume that workers are oblivious of the encoding/decoding process and simply perform matrix multiplication on the inputs they receive. Therefore, the $k^{\text{th}}$ worker  outputs $\widebf{C}_k = \widebf{A}_k\widebf{B}_k$ to the fusion node where $\widebf{C}_k \in \mathbb{F}^{\widetilde{\alpha}_k \times \widetilde{\gamma}_k}$. This model assumes that some workers are stragglers and may fail to respond. The fusion node downloads the responses from the non-straggling workers and attempts to recover the desired products $\{\mathbf{A}_i\mathbf{B}_j: (i,j) \in \mathcal{S}\}$ using a class of decoding functions (denoted $\mathbf{d}$).  We define $\mathbf{d}=\{d_{\mathcal{R}}: \mathcal{R}\subset[K]\}$ where $d_{\mathcal{R}}$ is the decoding function used when the set of non-straggling workers is $\mathcal{R} \subset [K]$. Therefore, we denote a VCDBMM code by the triple $(\mathbf{f},\mathbf{g},\mathbf{d})$.
	
	For convenience, we can define the VCDBMM problem using graph theoretic terms.  First, we observe that the desired computations $\mathcal{S}$ can be specified by using bipartite graphs. Let $G = ([L_A]\cup [L_B],\mathcal{S})$ where an edge exists between left vertex $i \in [L_A]$ and right vertex $j\in [L_B]$ if and only if $(i,j) \in \mathcal{S}$. As such, there is a one-to-one correspondence between the computation task and the bipartite graph which allows us to use them interchangeably. Similarly, we can define a VCDBMM problem by a bi-adjacency matrix of this bipartite graph. Finally, we define the degree of the $i^{\text{th}}$ left vertex as $d^A_i = |\{j: (i,j) \in \mathcal{S} \}|$ and the degree of the $j^{\text{th}}$ right vertex as $d^B_j = |\{i: (i,j) \in \mathcal{S} \}|$. In subsequent sections, we will be analyzing ensembles of VCDBMM problems. We define one ensemble to be an Erdős–Rényi model where each edge is included in the graph with probability $\lambda$, $0< \lambda <1$, independently from every other edge. Thus, we parameterize the ensemble as $V_{\lambda}(L_A,L_B)$. We define an additional model where each $\mathbf{A}_i$ matrix can be multiplied with at most $k$ $\mathbf{B}_i$ matrices. In terms of the bipartite graph, this ensemble first uniformly selects a degree for each left vertex from $1$ to $k$ and then randomly assigns its edges to the right vertices, avoiding the creation of parallel edges. We parameterize the ensemble as $V_k(L_A,L_B)$. These ensembles are representative of the different types of use cases for FCSA codes. 

	Now, we  define the major parameters for our system:
	\begin{itemize}
		\item \textit{Recovery Threshold}: Minimum integer $R$ where for any $\mathcal{R}\subset[K]$ such that $|\mathcal{R}| \geq R$ and for any realization of $\mathcal{A}$ and $\mathcal{B}$, then $d_{\mathcal{R}}(\{\widebf{C}_i: i\in \mathcal{R}\}) = \{\mathbf{A}_i\mathbf{B}_j: (i,j) \in \mathcal{S}\}$.
		\item \textit{Upload Cost}: $U_A$ and $U_B$ are the average number of symbols (from $\mathbb{F}$) sent from sources $\mathbb{A}$ and $\mathbb{B}$ to the workers, i.e., $	U_A =\frac{1}{K} \sum_{i=1}^{K}\widetilde{\alpha}_k\widetilde{\beta}_k, 
		U_B =\frac{1}{K} \sum_{i=1}^{K}\widetilde{\beta}_k\widetilde{\gamma}_k.$
		\item \textit{Download Cost}: $D_C$ is the worst-case average number of symbols (from $\mathbb{F}$) transmitted from the workers to the fusion node, i.e., $D_C=\max_{\mathcal{R}, \mathcal{R}\subset[K], |\mathcal{R}|=R}\frac{\sum_{k\in\mathcal{R}}\widetilde{\alpha}_k\widetilde{\gamma}_k}{R}$ where $R$ is the recovery threshold.
		\item \textit{Encoding Complexity}: $C_A$ and $C_B$ are the encoding complexities to generate the outputs of sources $\mathbb{A}$ and $\mathbb{B}$, respectively.
		\item \textit{Worker Complexity}: $C_w$ is the average computational complexity across workers. 
		\item \textit{Decoding Complexity}: $C_d$ is the decoding complexity of extracting the desired computations from the $R$ responses. 
	\end{itemize}
	
	\vspace{-0.5em}
	\subsection{Previous Results on Interpolating Rational Functions}
	
	For our coding schemes, we will rely on the following lemma about rational function interpolation as a fundamental building block of our code construction, taken from \cite{gasca1989computation}.
	\begin{lemma}\label{lemma:cauchy_vandermonde}
		(\cite{gasca1989computation}) Let $f_{1}, f_{2}, \cdots, f_{N}, x_1,x_2,\cdots,x_K$ be $N+K$ distinct elements of $\mathbb{F}$, with $|\mathbb{F}|\geq N+K$. Each element $f_i$ has an associated multiplicity of $u_i \in \mathbb{Z}_{\geq}$ where $u_1+u_2 + \cdots+ u_N = M$. Let $M+1 \leq K$. Then, the coefficients $e_{i,j}$ of the following function can be interpolated from the function outputs of the $K$ evaluation points (i.e., $\{F(x_i):i\in[K]\}$):
		\begin{equation}
		F(z) = \sum_{i=1}^{N}\sum_{j=1}^{u_i}\frac{e_{i,j}}{(z-f_i)^j} + \sum_{j=0}^{K-M-1}e_{0,j}z^j. \label{eq:cauchy_van_function}
		\end{equation}
	\end{lemma}

	\subsection{Relevant Constructions}\label{subsec:relevant_cons}
	In this section, we provide a brief overview of three relevant code constructions that are achievable strategies for the VCDBMM problem. We consider these three schemes as a baseline for comparison to our FCSA codes.

	\subsubsection{Polynomial Codes}\label{subsubsec:poly_codes}
	Polynomial codes \cite{Yu2017PolynomialCA} are codes based on matrix partitioning to efficiently distribute the computation of multiplying a single pair of large matrices $\mathbf{A}$ and $\mathbf{B}$. Polynomial codes can be used to solve the VCDBMM problem by having the Polynomial code calculate all pairs of $\mathbf{A}_i$ and $\mathbf{B}_j$ which would require a recovery threshold of $L_AL_B$. While such a computation strategy would clearly not be optimal for a general VCDBMM problem, it offers a good upper bound on achievable schemes and, in fact, performs better than some of the other codes for highly dense VCDBMM problems.

	\subsubsection{Lagrange Coded Computing Codes}
	Lagrange Coded Computing (LCC) codes \cite{Yu2019LagrangeCC} are a class of codes designed for distributed batch multivariate polynomial evaluation based on encoding using Lagrange polynomials. For the purposes of this paper, we note that LCC codes can solve the distributed Batch Matrix Multiplication problem of pairwise multiplying two sets of $L$ matrices $\{\mathbf{A}_1,\mathbf{A}_2,\dots,\mathbf{A}_L\}$ and $\{\mathbf{B}_1,\mathbf{B}_2,\dots,\mathbf{B}_L\}$ into $\{\mathbf{A}_1\mathbf{B}_1,\mathbf{A}_2\mathbf{B}_2,\dots,\mathbf{A}_L\mathbf{B}_L\}$ with a recovery threshold of $2L-1$. While this is a special case of the VCDBMM problem, LCC codes can be used to solve the general VCDBMM problem by duplicating $\mathbf{A}_i$ and $\mathbf{B}_j$ matrices to create two lists of matrices with $|\mathcal{S}|$ matrices which requires a recovery threshold of $2|\mathcal{S}| - 1$.

	\subsubsection{Cross-Subspace Alignment Codes}
	We now describe a class of codes that are a precursor to FCSA codes: Cross-Subspace Alignment (CSA) codes \cite{jia2021cross}. The main idea behind CSA codes is that the desired computations are encoded in the rational terms $\frac{1}{(f_i - x_k)^r}$ while all undesired computations (interference terms) are encoded in the polynomial terms $x^r_k$. For future comparison to FCSA, we now provide an example of CSA codes. Similar to LCC codes, CSA codes are designed for the Batch Matrix Multiplication problem. As such, let $\mathcal{A} = \{\mathbf{A}_1,\mathbf{A}_2\}$, $\mathcal{B} = \{\mathbf{B}_1,\mathbf{B}_2\}$, and $\mathcal{S} = \{(1,1),(2,2)\}$, i.e., the goal is to compute $\mathbf{A}_1\mathbf{B}_1$ and $\mathbf{A}_2\mathbf{B}_2$. Let $f_1$ and $f_2$ be distinct elements in $\mathbb{F}$. Consider the following encoding polynomials:
	\begin{align}
	\mathbf{A}(x) &= (x-f_1)(x-f_2)(\frac{\mathbf{A}_1}{x-f_1}+\frac{\mathbf{A}_2}{x-f_2}) = (x-f_2)\mathbf{A}_1 + (x-f_1)\mathbf{A}_2, \\
	\mathbf{B}(x) &=\frac{\mathbf{B}_1}{x-f_1}+\frac{\mathbf{B}_2}{x-f_2}.
	\end{align}
	The transmitted matrices $\widebf{A}_k$ and $\widebf{B}_k$ are different evaluations of the previous polynomials, i.e., $\widebf{A}_k = \mathbf{A}(x_k)$ and $\widebf{B}_k = \mathbf{B}(x_k)$ for some $x_k$ distinct from $f_1$ and $f_2$. As such, the result from a worker is an evaluation of the following polynomial
	\begin{align}
	\mathbf{A}(x)\mathbf{B}(x) &= \frac{\mathbf{A}_1\mathbf{B}_1(x-f_2)}{x-f_1}+\frac{\mathbf{A}_2\mathbf{B}_2(x-f_1)}{x-f_2} \\
	&= \frac{\mathbf{A}_1\mathbf{B}_1(f_1-f_2)}{x-f_1} + \frac{\mathbf{A}_2\mathbf{B}_2(f_2-f_1)}{x-f_2} + \mathbf{A}_1\mathbf{B}_1 + \mathbf{A}_2\mathbf{B}_2.
	\end{align}
	 By Lemma \ref{lemma:cauchy_vandermonde}, $\mathbf{A}(x)\mathbf{B}(x)$ can be interpolated using $3$ evaluations and the desired computations $\mathbf{A}_1\mathbf{B}_1$ and $\mathbf{A}_2\mathbf{B}_2$ can be easily extracted from the coefficients of the rational terms. Thus, for this example, the recovery threshold of CSA codes is $3$. In general, given $L$ pairs of matrices, CSA codes have a recovery threshold\footnote{CSA codes can have a smaller recovery threshold if the matrices are allowed to be grouped up and each group is individually encoded using MDS codes which changes the communication and computation complexity. To provide a fair comparison, we shall focus on comparing CSA codes with FCSA codes when using equivalent communication and computation complexity.} of $2L-1$. Since CSA codes solve the Batch Matrix Multiplication problem, we use a similar technique as for LCC codes and duplicate the matrices to get a recovery threshold of $2|\mathcal{S}| - 1$.

	Using the previously mentioned code constructions, we thus have an upper bound on the achievable recovery threshold of the VCDBMM problem with $\min(L_AL_B,2|\mathcal{S}| - 1)$. In the next section, we provide the construction for our new FCSA codes which significantly improve upon this recovery threshold. 
	\vspace{-1em}
	\section{Flexible Cross-Subspace Alignment codes } \label{sec:fcsa}
	\vspace{-0.3em}
	
	In this section, we present our Flexible Cross-Subspace Alignment (FCSA) codes to solve the VCDBMM problem. We shall first describe important constructs for our code design. We will subsequently demonstrate a motivating example using these constructs. Finally, we will provide our code construction of FCSA codes. 
	
		\subsection{Important Constructs}
	Before describing the code construction, we define important constructs for our code. The first construct will be known as the task assignment.
	
	\begin{definition}\label{def:task}
		For a given computation list $\mathcal{S}$, we define a task assignment $\mathcal{Q}$ as a set of tuples $\mathcal{Q} = \{(\mathcal{L}^q_A,\mathcal{L}^q_B)\}_{q=1}^{|\mathcal{Q}|}$ where $\mathcal{L}^q_A \subseteq [L_A],\mathcal{L}^q_B \subseteq [L_B]$ and define $L^q_A = |\mathcal{L}^q_A|$ and $L^q_B = |\mathcal{L}^q_B|$. The set $\mathcal{Q}$ is chosen such that $(L^q_A\times L^q_B) \cap (L^{q'}_A\times L^{q'}_B) = \emptyset$ for all distinct $q,q' \in [|\mathcal{Q}|]$ and for every $s \in \mathcal{S}$ there exists exactly one $q \in [|\mathcal{Q}|]$ such that $s \in \mathcal{L}^q_A\times \mathcal{L}^q_B$. We shall refer to each tuple as a task group.
	\end{definition}
	Intuitively, the task assignments are grouping computations such that each computation within a task assignment $\mathcal{L}^q_A \times \mathcal{L}^q_B$ will be recovered. Essentially, each task group is similar to problem space (i) where we will acquire all matrix products between two sets. Since every $s \in \mathcal{S}$ is part of some task group, we will be able to recover the desired computations. The next construct we define is the power assignment.
	
	\begin{definition}\label{def:power}
		Given a task assignment $\mathcal{Q}$, we define a power assignment $\mathcal{P}$ as a tuple of vectors $\mathcal{P}= \{(P^{A,q},P^{B,q})\}_{q=1}^{|\mathcal{Q}|}$ where $P^{A,q} \in (\mathcal{Z}_{\geq})^{L_A}$ and $P^{B,q} \in (\mathcal{Z}_{\geq})^{L_B}$ for all $q \in [|\mathcal{Q}|]$. The power assignment $\mathcal{P}$ must satisfy the following constraints:
		\begin{enumerate}
			\item $P^{A,q}_i=0$ if $i \notin \mathcal{L}^{q}_A$ and $P^{B,q}_j = 0$ if $j \notin \mathcal{L}^{q}_B$ for all $q\in [|\mathcal{Q}|]$. \label{def:power_con_1}
			\item For each $q \in [|\mathcal{Q}|]$, exactly one of the following is true: \label{def:power_con_2}
			\begin{enumerate}[label=\roman*)]
				\item $P^{A,q}_i = L^q_AL^q_B - x+1$ for some $x\in [L^q_A]$ and $P^{B,q}_j =yL^q_A$ for some $y \in [L^q_B]$ for all $i \in \mathcal{L}^q_A$, $j\in \mathcal{L}^q_B$. \label{con:power_a}
				\item $P^{B,q}_i = L^q_AL^q_B - x+1$ for some $x\in [L^q_B]$ and $P^{A,q}_j =yL^q_B$ for some $y \in [L^q_A]$  for all $i \in \mathcal{L}^q_A$, $j\in \mathcal{L}^q_B$. \label{con:power_b}
			\end{enumerate}
			\item For each $q \in [|\mathcal{Q}|]$, all non-zero elements in $P^{A,q}$ are distinct from each other and all non-zero elements in $P^{B,q}$ are distinct from each other. \label{def:power_con_3}
		\end{enumerate}
	\end{definition}
	
	The power assignment is defined in such a way as to prevent matrix products from interfering with each other in our coding scheme, as will be shortly shown. We prove the following key facts about the power assignment.
	\begin{lemma} \label{lemma:power}
		From Definition \ref{def:power}, we can prove the following facts about the power assignment:
		\begin{enumerate}[labelsep=0.5mm]
			\item For all $q\in [|\mathcal{Q}|]$, $P^{A,q}_i+P^{B,q}_j > L^q_AL^q_B$ for $(i,j) \in \mathcal{L}^q_A\times \mathcal{L}^q_B$  and $P^{A,q}_i+P^{B,q}_j \leq L^q_AL^q_B$ for  $(i,j) \notin \mathcal{L}^q_A\times \mathcal{L}^q_B$. \label{lemma:power_state_1}
			\item For $(i,j) \in \mathcal{L}^q_A\times \mathcal{L}^q_B$, $P^{A,q}_i+P^{B,q}_j -L^q_AL^q_B \in [L^q_AL^q_B]$.\label{lemma:power_state_2}
			\item For distinct pairs $(i,j),(k,l) \in  \mathcal{L}^q_A\times \mathcal{L}^q_B$, $P^{A,q}_i+P^{B,q}_j \neq P^{A,q}_k+P^{B,q}_l $.\label{lemma:power_state_3}
		\end{enumerate}
	\end{lemma}
	\begin{proof}
		We will prove each part in order. Without loss of generality, we assume that condition \ref{def:power_con_2}.\ref{con:power_a} holds in Definition \ref{def:power} due to the symmetry of conditions \ref{def:power_con_2}.\ref{con:power_a} and \ref{def:power_con_2}.\ref{con:power_b}. Additionally, all the statements are only for the values within one group in $\mathcal{Q}$. Thus, assume that we focus on the $q^{\text{th}}$ group in $\mathcal{Q}$.
		\begin{enumerate}[labelsep=0.5mm]
			\item First, we note that $P^{A,q}_i \leq L^q_AL^q_B$ and $P^{B,q}_i \leq L^q_AL^q_B$. Next, consider the case $(i,j) \notin \mathcal{L}^q_A\times \mathcal{L}^q_B$. By Definition \ref{def:power}.\ref{def:power_con_1}), we have that at least $P^{A,q}_i$ or $P^{B,q}_i$ is zero. Thus, $P^{A,q}_i+P^{B,q}_j \leq L^q_AL^q_B$. Now, consider the  case $(i,j) \in \mathcal{L}^q_A\times \mathcal{L}^q_B$. Recall that $P^{A,q}_i =  L^q_AL^q_B - x+1$ for some $x\in [L^q_A]$ which implies $P^{A,q}_i \geq L^q_AL^q_B - L^q_A+1$. Additionally, $P^{B,q}_j =yL^q_A$ for some $y \in [L^q_B]$ which implies $P^{B,q}_j \geq L^q_A$. Hence, $P^{A,q}_i+P^{B,q}_j \geq L^q_AL^q_B - L^q_A+1 +L^q_A = L^q_AL^q_B +1 > L^q_AL^q_B$.
			\item  Let $(i,j) \in  \mathcal{L}^q_A\times \mathcal{L}^q_B$. Due to Lemma \ref{lemma:power}.\ref{lemma:power_state_1}), we know that $P^{A,q}_i+P^{B,q}_j - L^q_AL^q_B > 0$. Thus, we only have to prove that $P^{A,q}_i+P^{B,q}_j \leq 2 L^q_AL^q_B$. This is straightforward to prove by following the similar logic used earlier in the lemma to arrive at $P^{A,q}_i \leq  L^q_AL^q_B$ and $P^{B,q}_j \leq L^q_AL^q_B$ which completes this part of the proof.
			\item To prove this fact, observe that the difference between the highest and lowest value that $P^{A,q}_i$ can take is $L^q_A-1$. Additionally, note that $P^{B,q}_i$ takes values in the set $\{L^q_A,2L^q_A, \dots, L^q_BL^q_A\}$ and thus the difference between any two values $P^{B,q}_i$ takes is at least $L^q_A$. Thus, $P^{A,q}_i+P^{B,q}_j \neq P^{A,q}_k+P^{B,q}_l$.
		\end{enumerate}
	\end{proof}
	
	Note that Lemma \ref{lemma:power}.\ref{lemma:power_state_2}) and \ref{lemma:power}.\ref{lemma:power_state_3}) together equivalently mean that the set $\{P^{A,q}_i+P^{B,q}_j -L^q_AL^q_B : (i,j) \in \mathcal{L}^q_A\times \mathcal{L}^q_B \}$ is a permutation of the set $[L^q_AL^q_B]$. This interpretation will be useful when discussing the optimization of the power assignment. In the following section, we shall provide a motivating example of how to use task and power assignments to create FCSA codes.

	\subsection{A Motivating Example}\label{subsec:motivating_ex}
	To provide some insight into how FCSA codes are constructed, we shall work through a small motivating example. Consider the VCDBMM problem with $\mathcal{S}= \{(1,1),(1,2),(2,2),(2,3)\}$. This problem can be summarized by the bi-adjacency matrix in Fig. \ref{fig:small_ex} which also provides the task and power assignment.

	\begin{figure}[t]
	\begin{center}
	\begin{tabular}{l|c|| c | c | c |}
		\cline{3-5}
		\multicolumn{1}{c}{} &\multicolumn{1}{c||}{}  &	$\mathbf{B}_1$ & $\mathbf{B}_2$ & $\mathbf{B}_3$  \\\cline{2-5}
		& \diagbox{$P^{A,q}$}{\vspace{-1em}\\ $P^{B,q}$}&	 
		\textcolor{group3}{2},\textcolor{group2}{0}&  
		\textcolor{group3}{1},\textcolor{group2}{1}& 
		\textcolor{group3}{0},\textcolor{group2}{2}\\\hline\hline
		\multicolumn{1}{|c|}{$\mathbf{A}_1$} & \textcolor{group3}{2},\textcolor{group2}{0}
		&\tikzmark{top left 1}1 & 1\tikzmark{bottom right 1} & 0 \\\hline
		\multicolumn{1}{|c|}{$\mathbf{A}_2$} & \textcolor{group3}{0},\textcolor{group2}{2}
	 & 0 &\tikzmark{top left 2}1 & 1 \tikzmark{bottom right 2} \\	\hline
	\end{tabular}
	\DrawBox[ultra thick, group3]{top left 1}{bottom right 1}
	\DrawBox[ultra thick, group2]{top left 2}{bottom right 2}
			\tikz[overlay,remember picture] {
		\draw[<-,>=stealth] (bottom right 1) -- ++(0:3cm) node[draw, fill=white,text= group3, pos=1]{\scriptsize Group 1};
		\draw[<-,>=stealth] (bottom right 2) -- ++(0:2cm) node[draw, fill=white,text = group2, pos=1]{\scriptsize Group 2};
	}
	\end{center}
	\caption{A small example of FCSA code with a given task and power assignment. Task groupings are delineated by different colors such that if an element in row $i$ (column $j$) belongs to a color, then $i \in \mathcal{L}^q_A (j \in \mathcal{L}^q_B)$ for the  group represented by the color. Results in a recovery threshold of $R_{FCSA} = 5$.} \label{fig:small_ex}
	\end{figure}

	Let the elements $f_1,f_2 \in \mathbb{F}$ be associated with task group 1 and 2, respectively. Using the power assignments in Fig. \ref{fig:small_ex}, we create the following encoding polynomials:
	\begin{equation}
	\begin{split}
	\Theta(x) &= (x-f_1)^{L^1_AL^1_B} (x-f_2)^{L^2_AL^2_B}= (x-f_1)^2(x-f_2)^2,\\
	\mathbf{A}(x) &=\Theta(x) \sum_{i=1}^{L_A}\frac{\mathbf{A}_i}{(x-f_1)^{P^{A,1}_i} (x-f_2)^{P^{A,2}_i}} = \Theta(x) \left(\frac{\mathbf{A}_1}{(x-f_1)^2} +  \frac{\mathbf{A}_2}{(x-f_2)^2}\right), \\
	\mathbf{B}(x) &= \sum_{j=1}^{L_B}\frac{\mathbf{B}_j}{(x-f_1)^{P^{B,1}_i} (x-f_2)^{P^{B,2}_i}} =\frac{\mathbf{B}_1}{(x-f_1)^2} +  \frac{\mathbf{B}_2}{(x-f_1)(x-f_2)} +  \frac{\mathbf{B}_3}{(x-f_2)^2}. \\
	\end{split}
	\end{equation}
	Recall that $L^q_A$ and $L^q_B$ are the number of $\mathbf{A}_i$ and $\mathbf{B}_i$ matrices in group $q$, respectively. Now, we transmit to the $k^{\text{th}}$ worker the evaluations $\mathbf{A}(x_k)$ and $\mathbf{B}(x_k)$ for a unique element $x_k \in \mathbb{F}$ different from $f_1$ and $f_2$. Since the worker multiplies the matrices it receives, the output of each worker is $\mathbf{A}(x_k)\mathbf{B}(x_k)$. We now drop the subscript on $x_k$ for convenience. Note that the output of a worker is an evaluation of the matrix function  $\mathbf{A}(x)\mathbf{B}(x)$ which can be expanded as follows: 
	\begin{align}
		&\mathbf{A}(x)\mathbf{B}(x) \nonumber \\
		&= \mathbf{A}_1\mathbf{B}_1\left(\frac{\Theta(x)}{(x-f_1)^4}\right) + \mathbf{A}_2\mathbf{B}_1\left(\frac{\Theta(x)}{(x-f_1)^2(x-f_2)^2}\right)  \nonumber\\
		&+ \mathbf{A}_1\mathbf{B}_2\left(\frac{\Theta(x)}{(x-f_1)^3(x-f_2)}\right)  +\mathbf{A}_2\mathbf{B}_2\left(\frac{\Theta(x)}{(x-f_1)(x-f_2)^3}\right)  \nonumber\\
		&+ \mathbf{A}_1\mathbf{B}_3\left(\frac{\Theta(x)}{(x-f_1)^2(x-f_2)^2}\right)  +\mathbf{A}_2\mathbf{B}_3\left(\frac{\Theta(x)}{(x-f_2)^4}\right) \\
		&= \mathbf{A}_1\mathbf{B}_1\frac{(x-f_2)^2}{(x-f_1)^2} + \mathbf{A}_2\mathbf{B}_1  
		+ \mathbf{A}_1\mathbf{B}_2\frac{(x-f_2)}{(x-f_1)}  +\mathbf{A}_2\mathbf{B}_2\frac{(x-f_1)}{(x-f_2)} 
		+ \mathbf{A}_1\mathbf{B}_3 +\mathbf{A}_2\mathbf{B}_3\frac{(x-f_1)^2}{(x-f_2)^2}.  \label{eq:small_ex_1}
	\end{align}

	Note that the coefficients of the rational terms are the matrix products in the task assignments. Additionally, note that the coefficients for the terms with $\frac{1}{x-f_q}$ are the matrix products associated with group $q$. This is due to Lemma \ref{lemma:power} by noting that the powers of the terms in the denominators of the rational function associated with $\mathbf{A}_i\mathbf{B}_j$ are $P^{A,q}_i+P^{B,q}_i$ for that associated group element $f_q$. Thus, for the matrix products not in the task assignment, the denominator has terms with smaller powers than the terms in $\Theta(x)$ and, thus, become polynomial terms. Moreover, the matrix products within group $q$ are each associated with a different power for their denominator. For example, the following terms are associated with group $1$:  $\mathbf{A}_1\mathbf{B}_1\frac{(x-f_2)^2}{(x-f_1)^2} +
	 \mathbf{A}_1\mathbf{B}_2\frac{(x-f_2)}{(x-f_1)}$.
	 
	We can now continue to simplify to get
	\begin{equation}
	\begin{split}
		\mathbf{A}(x)\mathbf{B}(x) &=  \mathbf{A}_1\mathbf{B}_1\left(\frac{e_{1,1,1}}{(x-f_1)^2}+\frac{e_{1,1,2}}{(x-f_1)}+ z_{1,1,1}\right)  +\mathbf{A}_1\mathbf{B}_2\left(\frac{e_{1,2,1}}{(x-f_1)} +z_{1,2,1}\right)  \\
		&+ \mathbf{A}_2\mathbf{B}_2\left(\frac{e_{2,2,1}}{(x-f_2)}+z_{2,2,1}\right)
		+\mathbf{A}_2\mathbf{B}_3\left(\frac{e_{2,3,1}}{(x-f_2)^2}+\frac{e_{2,3,2}}{(x-f_2)}+z_{2,3,1}\right)   \\
		&+	\mathbf{A}_2\mathbf{B}_1 +  \mathbf{A}_1\mathbf{B}_3, 
	\end{split} \label{eq:small_ex_2}
	\end{equation}
	\begin{align}
= \frac{e_{1,1,1}\mathbf{A}_1\mathbf{B}_1}{(x-f_1)^2}
+\frac{e_{1,2,1}\mathbf{A}_1\mathbf{B}_2+e_{1,1,2}\mathbf{A}_1\mathbf{B}_1}{(x-f_1)} 
+\frac{e_{2,3,1}\mathbf{A}_2\mathbf{B}_3}{(x-f_2)^2}
+\frac{e_{2,2,1}\mathbf{A}_2\mathbf{B}_2+e_{2,3,2}\mathbf{A}_2\mathbf{B}_3}{(x-f_2)}  +	\mathbf{I}_1, \label{ex:small_ex_2}
	\end{align}
	where Eq. \eqref{eq:small_ex_2} comes from partial fraction decomposition and $\mathbf{I}_1 = \mathbf{A}_2\mathbf{B}_1 +  \mathbf{A}_1\mathbf{B}_3$. Note that $e_{1,1,1},e_{1,2,1},e_{2,2,1},e_{2,3,1}$ are all non-zero by partial fraction decomposition otherwise we could have simplified the equations further. 
	 By Lemma \ref{lemma:cauchy_vandermonde}, this matrix function can be interpolated from $5$ evaluation points and we can recover the following coefficients:
 	\begin{equation}
	 \begin{split}
	 \{e_{1,1,1}\mathbf{A}_1\mathbf{B}_1,e_{1,2,1}\mathbf{A}_1\mathbf{B}_2+e_{1,1,2}\mathbf{A}_1\mathbf{B}_1,
	 e_{2,3,1}\mathbf{A}_2\mathbf{B}_3,
	 e_{2,2,1}\mathbf{A}_2\mathbf{B}_2+e_{2,3,2}\mathbf{A}_2\mathbf{B}_3\}.
	 \end{split} 
	 \end{equation}
	 We observe that the matrix products within a group are now encoded into a triangular system of linear equations. Since the leading terms $e_{1,1,1},e_{1,2,1},e_{2,2,1},e_{2,3,1}$ are non-zero, the triangular systems are invertible. Hence, we are able to recover $\mathbf{A}_1\mathbf{B}_1,\mathbf{A}_1\mathbf{B}_2,\mathbf{A}_2\mathbf{B}_3, \mathbf{A}_2\mathbf{B}_2$ which are exactly the desired terms in the computation list $\mathcal{S}$. The recovery threshold for this scheme is $5$ since only $5$ worker outputs are necessary to interpolate $\mathbf{A}(x)\mathbf{B}(x)$. In comparison, the schemes discussed in Section \ref{subsec:relevant_cons} require a recovery threshold of $\min(2|\mathcal{S}|-1,L_AL_B) = \min(7,6) = 6$.
	
	For a more complex example, please refer to Fig. \ref{fig:example} where we demonstrate an example task and power assignment that results in a smaller recovery threshold in comparison to codes discussed in Section \ref{subsec:relevant_cons}. In the next section, we provide the full code construction for FCSA codes.

	\subsection{Main Theorem}
	Now, we are ready to state and prove our main result. 
	\begin{figure}
		\centering
		\includegraphics[width=\linewidth]{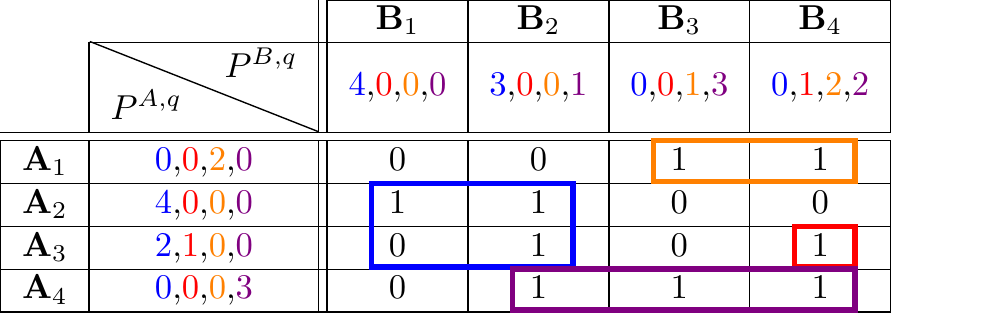}
		\caption{An example of a VCDBMM problem with a task and power assignment.
			Task groupings are delineated by different colors such that if an element in row $i$ (column $j$) belongs to a color, then $i \in \mathcal{L}^q_A (j \in \mathcal{L}^q_B)$ for the  group represented by the color. The power assignments are colored according to the task group it belongs to. By Theorem \ref{theorem:fcsa_codes}, $R_{FCSA}=15$. In contrast, the best recovery threshold using the schemes in Section \ref{subsec:relevant_cons} is $\min(16,17) = 16$. }  \label{fig:example}
		\vspace{-0.5cm}
	\end{figure}
	
	\begin{theorem}\label{theorem:fcsa_codes} (Achievability of FCSA codes)
		For a given computation list $\mathcal{S}$, assume that a valid task assignment $\mathcal{Q}$ and associated power assignment $\mathcal{P}$ are provided. Let $K$ be the number of workers. Then, assuming that $|\mathbb{F}| > |\mathcal{Q}| + K  $,  FCSA codes achieve the following:
		\hspace{-0.5em}
	
			\begin{align}
			&\text{Recovery Threshold: }  R_{FCSA} = 2\sum_{q=1}^{|\mathcal{Q}|}L^q_AL^q_B 
			- \min_{i\in [L_A]}\left(\sum_{q=1}^{|\mathcal{Q}|}P^{A,q}_{i}\right) 
			- \min_{j\in[L_B]}\left(\sum_{q=1}^{|\mathcal{Q}|}P^{B,q}_{j}\right) + 1, \\
			&\text{Upload Costs: }(U_A,U_B) =  \left(\alpha\beta,\beta\gamma\right),\\
			&\text{Download Cost: }  D_C = \alpha\gamma ,\\
			&\text{Encoding Complexities: }  \begin{cases}
			C_A  &=  \mathcal{O}\left(\alpha\beta KL_A\right),\\
			C_B  &= \mathcal{O}\left(\beta\gamma KL_B\right), \\
			\end{cases} \\
			&\text{Worker Complexity: }  C_w = \mathcal{O}\left(\alpha\beta\gamma\right) ,\\
			&\text{Decoding Complexity: }  C_d =  \widetilde{\mathcal{O}}(\alpha\gamma(R_{FCSA}\log^2(R_{FCSA}) + \sum_{q=1}^{|\mathcal{Q}|}(L^q_AL^q_B)^2)).
			\end{align}
	\end{theorem}

	\begin{proof}
		
	First, let $f_1,f_2,\dots, f_{|\mathcal{Q}|},x_1,x_2,\dots,x_K$ be distinct elements from $\mathbb{F}$. For all $i \in [L_A]$, $j\in [L_B]$, we define 
	
	\vspace{-1em}
	\begin{equation}
	a_i(x) = \prod_{q=1}^{|\mathcal{Q}|}(x-f_{q})^{P^{A,q}_{i}},
	b_j(x) =\prod_{q=1}^{|\mathcal{Q}|}(x-f_{q})^{P^{B,q}_{j}} .
	\end{equation}
	
	Note that the degrees of the polynomials of $a_i(x)$ and $b_j(x)$ are $\deg(a^i_r(x)) = \sum_{q=1}^{|\mathcal{Q}|}P^{A,q}_{i}$ and $\deg(b^j_r(x)) = \sum_{q=1}^{|\mathcal{Q}|}P^{B,q}_{j}$. Additionally, note that $a_i(x)$ and $b_j(x)$ share exactly one root $f_{q}$ if $(i,j) \in L^q_A\times L^q_B$, i.e., $P^{A,q}_i$ and $P^{B,a}_j$ are both non-zero; otherwise $a_i(x)$ and $b_j(x)$ do not share a root. If $a_i(x)$ and $b_j(x)$ share a root, we shall denote this shared root as $f_{i,j}$.
	
	Given $q \in [|\mathcal{Q}|] $ and  $(i,j) \in \mathcal{L}^q_A\times \mathcal{L}^q_B$, we define 
	\begin{align}
	\Delta_{i,j,q}(x) =  \prod_{p\in[|\mathcal{Q}|]: p \neq q} (x+(f_{q}-f_{p}))^{L^p_AL^p_B-P^{A,p}_i-P^{B,p}_j}.
	\end{align}
	Note that this product is over all the groups in which $(i,j)$ does not belong to, i.e., $(i,j) \notin \mathcal{L}^p_A\times \mathcal{L}^p_B$.
	Additionally, observe that for $p \neq q$ then either $P^{A,p}_i =0$ or $P^{B,p}_i =0$ which guarantees, due to Lemma \ref{def:power}, that 
	$P^{A,p}_i + P^{B,p}_i \leq L^p_AL^p_B$ 
	for all $(i,j,p)$ tuples used in $\Delta_{i,j,q}(x)$. Thus, $\Delta_{i,j,q}(x)$ must be a polynomial and we can write $\Delta_{i,j}(x) = \sum_{l=0}^{D_{i,j,q}}\zeta_{i,j,q,l}x^l$ where 
	$D_{i,j,q}=\sum_{p=1, p\neq q}^{|\mathcal{Q}|}L^p_AL^p_B-P^{A,p}_i-P^{B,p}_j$ 
	is the degree of $\Delta_{i,j,q}(x)$ and $\zeta_{i,j,q,l} \in \mathbb{F}$ are the coefficients of $\Delta_{i,j,q}(x)$.
	
	We define
	\vspace{-0.3em}
	\begin{equation}
	\Theta(x) = \prod_{q=1}^{|\mathcal{Q}|}(x-f_{q})^{L^q_AL^q_B} .
	\end{equation}
	\vspace{-0.3em}
	Note that $\Theta(x)$ has degree $\sum_{q=1}^{|\mathcal{Q}|}L^q_AL^q_B$. Sources $\mathbb{A}$ and  $\mathbb{B}$  use the following  polynomials, respectively, to encode the matrices
	
	\vspace{-0.5em}
	\begin{equation}
	\mathbf{A}(x) = \Theta(x) \sum_{i=1}^{L_A}\frac{\mathbf{A}_i}{a_i(x)} , \mathbf{B}(x) = \sum_{j=1}^{L_B}\frac{\mathbf{B}_j}{b_i(x)} . \label{eq:encoding}
	\end{equation}
	\vspace{-0.5em}
	
	As such, for the $k^{\text{th}}$ worker node, the source nodes transmit the evaluations of the encoding functions at $x_k$, i.e., $\widebf{A}_k = \mathbf{A}(x_k)$ and $\widebf{B}_k = \mathbf{B}(x_k)$. Thus, the worker output is $\widebf{C}_k = \mathbf{A}(x_k)\mathbf{B}(x_k)$. We now drop the subscript on $x_k$ for ease of notation. 
	The resulting computation from a worker is then an evaluation of the following function
	
	\vspace{-0.6em}
	\begin{equation}
	\mathbf{A}(x)\mathbf{B}(x) = \sum_{i=1}^{L_A}\sum_{j=1}^{L_B}\frac{\Theta(x) }{a_i(x)b_j(x)}\mathbf{A}_i\mathbf{B}_j 
	\end{equation} 
	\vspace{-0.2cm}
	\begin{subequations}\label{eq:inter_all_pre}
		\begin{equation}
		\begin{split}
		&=  \sum_{(i,j) \notin  \cup^{|\mathcal{Q}|}_{q=1} \mathcal{L}^q_A\times \mathcal{L}^q_B }   \mathbf{A}_i\mathbf{B}_j
		 \prod_{q\in[|\mathcal{Q}|]}
		(x-f_{q})^{L^q_AL^q_B-P^{A,q}_i-P^{B,q}_j} \label{eq:inter__pre_a} 
		\end{split}
		\end{equation}
		\vspace{-0.5cm}
		\begin{equation} 
		\begin{aligned}
		&+  \sum_{q\in[|\mathcal{Q}|]} \sum_{(i,j) \in \mathcal{L}^q_A\times \mathcal{L}^q_B }   \bigg( \prod_{p\in[|\mathcal{Q}|]: p\neq q }(x-f_{p})^{L^p_AL^p_B-P^{A,p}_i-P^{B,p}_j} 
		\bigg)\frac{\mathbf{A}_i\mathbf{B}_j }{(x-f_{q})^{P^{A,q}_i+P^{B,q}_j - L^q_AL^q_B}}  \label{eq:inter__pre_b} 
		\end{aligned}
		\end{equation}
	\end{subequations}

	\begin{subequations}\label{eq:inter_all}
		\begin{equation}
		\begin{split}
		&=  \sum_{(i,j) \notin  \mathcal{E}}   \mathbf{A}_i\mathbf{B}_j
		\prod_{q\in[|\mathcal{Q}|]}
		(x-f_{q})^{L^q_AL^q_B-P^{A,q}_i-P^{B,q}_j} \label{eq:inter_a} 
		\end{split}
		\end{equation}
		\vspace{-0.5cm}
		\begin{equation} 
		\begin{aligned}
		&+ \sum_{(i,j) \in \mathcal{E} }  \mathbf{A}_i\mathbf{B}_j  \frac{\Delta_{i,j,q}(x-f_{q_{ij}})}{(x-f_{q_{ij}})^{P_{i,j}}} \label{eq:inter_b} 
		\end{aligned}
		\end{equation}
	\end{subequations}
	 where the last equality comes from the definition of $\Delta_{i,j,q}(x)$,  $\mathcal{E} = \cup^{|\mathcal{Q}|}_{q=1} \mathcal{L}^q_A\times \mathcal{L}^q_B $, and $P_{i,j} = P^{A,q}_i+P^{B,q}_i -L^q_AL^q_B$ where $q$ is the group for which $(i,j) \in \mathcal{L}^q_A\times \mathcal{L}^q_B$. $P_{i,j}$ is well-defined since the summation in Eq. \eqref{eq:inter_b} is over $(i,j) \in \mathcal{E}$. Note that $\mathcal{S} \subseteq \mathcal{E}$ by the definition of task assignment.
	Observe that $P_{i,j} \geq 0$ in Eq. \eqref{eq:inter_b} by Lemma \ref{lemma:power}.\ref{lemma:power_state_1}).
	  Conceptually, the first (Eq. \eqref{eq:inter_a}) and second (Eq. \eqref{eq:inter_b}) summations correspond to the computations not assigned and assigned to the task assignment $\mathcal{Q}$, respectively. Our goal is to extract the matrix coefficients in Eq. \eqref{eq:inter_b}.
	It is clear that Eq. \eqref{eq:inter_a} is a polynomial in $x$ with degree $\max_{(i,j) \notin \mathcal{E}}(\sum_{q=1}^{|\mathcal{Q}|}L^q_AL^q_B - \sum_{q=1}^{|\mathcal{Q}|}P^{A,q}_{i} - \sum_{q=1}^{|\mathcal{Q}|}P^{B,q}_{j})$.
	
	Now, let us consider the sum Eq. \eqref{eq:inter_b}. First, we look at the terms in the summation. For convenience, we fix $(i,j) \in \mathcal{E}$ and let $q$ be the group such that $(i,j) \in \mathcal{L}^q_A\times \mathcal{L}^q_B$. Thus, we get
	\begin{align}
	&  \mathbf{A}_i\mathbf{B}_j  \frac{\Delta_{i,j,q}(x-f_{q})}{(x-f_{q})^{P_{i,j}}} \\
	&= \mathbf{A}_i\mathbf{B}_j  \Big(\frac{\zeta_{i,j,q,0}}{(x-f_{q})^{P_{i,j}}} +
	\cdots+ \frac{\zeta_{i,j,q,P_{i,j}-1}}{(x-f_{q})} 
	 +\sum_{k = P_{i,j}}^{D_{i,j,q}} \zeta_{i,j,q,k}(x-f_{i,j,q})^{k-P_{i,j}} \Big) \label{eq:expansion_t2}
	\end{align}
	where Eq. \eqref{eq:expansion_t2} comes from partial fraction decomposition. Note that the maximal value for $P_{i,j}$ is $L^q_AL^q_B$ by Lemma \ref{lemma:power}. We thus claim that Eq. \eqref{eq:inter_b} results in $\sum_{q=1}^{|\mathcal{Q}|} L^q_AL^q_B$ rational terms and a polynomial of degree  $\max_{(i,j) \in \mathcal{E}}(D_{i,j,q} -P_{i,j})=  \max_{(i,j) \in \mathcal{E}} (\sum_{q=1}^{|\mathcal{Q}|}L^q_AL^q_B - \sum_{q=1}^{|\mathcal{Q}|}P^{A,q}_{i} - \sum_{q=1}^{|\mathcal{Q}|}P^{B,q}_{j})$. 
	
	From Eq. \eqref{eq:expansion_t2} and Eq. \eqref{eq:inter_a}, we see that Eq. \eqref{eq:inter_all} creates 
	$ \sum_{q=1}^{|\mathcal{Q}|}L^q_AL^q_B - \min_{i\in [L_A]}\left(\sum_{q=1}^{|\mathcal{Q}|}P^{A,q}_{i}\right) - \min_{j\in [L_B]}\left (\sum_{q=1}^{|\mathcal{Q}|}P^{B,q}_{j}\right) + 1$ polynomial terms and $\sum_{q=1}^{|\mathcal{Q}|}L^q_AL^q_B$ rational terms. 
	
	Let $T = \sum_{q=1}^{|\mathcal{Q}|}L^q_AL^q_B - \min_{i\in [L_A]}\left(\sum_{q=1}^{|\mathcal{Q}|}P^{A,q}_{i}\right) - \min_{j\in [L_B]}\left (\sum_{q=1}^{|\mathcal{Q}|}P^{B,q}_{j}\right )$. Applying the expression from Eq. \eqref{eq:expansion_t2} and combining the two expressions in \eqref{eq:inter_all}, we arrive at the following equation for $\mathbf{A}(x)\mathbf{B}(x)$:
	\begin{equation}
	\mathbf{A}(x)\mathbf{B}(x) = \sum_{i=0}^{T}I_rx^{r} +  \hspace{-1.5em}\sum_{\substack{q\in |\mathcal{Q}| \\(i,j) \in \mathcal{L}^q_A\times \mathcal{L}^q_B}}  \hspace{-1.5em}\frac{1}{(x-f_{q})^{P_{i,j}}}\hspace{-1em}\sum_{\substack{(k,l) \in \mathcal{L}^q_A\times \mathcal{L}^q_B \\ P_{k,l} \geq P_{i,j}}}
	\hspace{-1.5em}\zeta_{k,l,q,P_{k,l}-P_{i,j} } \mathbf{A}_k\mathbf{B}_l  \label{eq:fcsa_t2_final}
	\end{equation}
	where $\{I_i: i \in \{0,1,\dots,T\}\}$  are arbitrary interference terms. Note that by Lemma \ref{lemma:power}, $P_{k,l} = P_{i,j}$ only when $k = i$ and $l = j$.
	
	By Lemma \ref{lemma:cauchy_vandermonde}, $\mathbf{A}(x)\mathbf{B}(x)$ can be interpolated from \begin{equation}
	R_{FCSA}=2\sum_{q=1}^{|\mathcal{Q}|}L^q_AL^q_B  
	- \min_{i\in [L_A]}\left(\sum_{q=1}^{|\mathcal{Q}|}P^{A,q}_{i}\right) 
	- \min_{j\in[L_B]}\left(\sum_{q=1}^{|\mathcal{Q}|}P^{B,q}_{j}\right) + 1
	\end{equation} evaluations. Thus, we can acquire the coefficients associated with the rational and polynomial terms. Now, we show that $\{\mathbf{A}_i\mathbf{B}_j: (i,j) \in S\}$ can be extracted from the coefficients of $\mathbf{A}(x)\mathbf{B}(x)$ which proves that $R_{FCSA}$ is the recovery threshold.
	
	We observe that in Eq. \eqref{eq:fcsa_t2_final} the coefficients of the rational terms with a pole at $f_{q}$ contain only the computations where $\{\mathbf{A}_i\mathbf{B}_j: (i,j) \in \mathcal{L}^q_A\times \mathcal{L}^q_B\}$ are involved. As such, we fix $q$ and focus on the subspace generated by the powers of $\frac{1}{(x-f_{q})}$. For convenience, let $d = L^q_AL^q_B$. We define the ordered index set $(i_1,j_1),(i_2,j_2),\dots,(i_d,j_{d})$ for the values of $\{P_{i,j}:(i,j) \in \mathcal{L}^q_A\times \mathcal{L}^q_B\}$ where $a \leq b \iff P_{i_a,j_a} < P_{i_b,j_b}$. Note that by Lemma \ref{lemma:power}, $P_{i,j}$'s take distinct values from $[L^q_AL^q_B]$ and, thus, $P_{i_{k+1},j_{k+1}} = P_{i_{k},j_{k}}+1$ with $P_{i_1,j_1} = 1$. Hence, we can write the rational terms associated with the root $f_{q}$ as 
	\begin{align}
	\sum_{k=1}^{d}\frac{1}{(x-f_{q})^{P_{i_k,j_k}}}\underbrace{\sum_{l=k}^{d}\zeta_{i_k,j_k,r,{l-k}}\mathbf{A}_{i_l}\mathbf{B}_{j_l}}_{\mathbf{Y}_{k}} .
	\end{align}
	
	Note that $\mathbf{Y}_k$ are the coefficients extracted from $\mathbf{A}(x)\mathbf{B}(x)$. In matrix notation, we write $\mathbf{Y}_k$ as 
	\vspace{-0.1cm}
	\begin{equation}
	\begingroup 
	\setlength\arraycolsep{2pt}
	\hspace{-0.2em}
	\begin{bmatrix}
	\mathbf{Y}_{1}\\
	\mathbf{Y}_{2}\\
	\vdots \\
	\mathbf{Y}_{d}
	\end{bmatrix} \hspace{-0.4em}=\hspace{-0.4em}
	\begin{bmatrix}
	\zeta_{{i_1,j_1,r,0}} & \zeta_{{i_2,j_2,r,1}} & \dots & \zeta_{{i_d,j_{d},r,d-1}}\\
	& \zeta_{{i_2,j_2,r,0}}  & \dots & \zeta_{{i_d,j_d,r,d-2}}\\
	&&\ddots & \vdots  \\
	&&& \zeta_{{i_{d},j_d,r,0}}\\
	\end{bmatrix}
	\hspace{-0.5em}
	\begin{bmatrix}
	\mathbf{A}^{i_1}_r\mathbf{B}^{j_1}_r\\
	\mathbf{A}^{i_2}_r\mathbf{B}^{j_2}_r\\
	\vdots \\
	\mathbf{A}^{i_d}_r\mathbf{B}^{j_d}_r\\
	\end{bmatrix}. \label{eq:upper_triangular}
	\endgroup
	\end{equation}
	To extract $\{\mathbf{A}_i\mathbf{B}_j:(i,j) \in \mathcal{L}^q_A\times \mathcal{L}^q_B\}$, we need to show that the upper-triangular matrix is invertible. It is sufficient to show that all elements along the diagonal are non-zero which is clearly true since $\zeta_{i,j,r,0} = \Delta_{i,j,r}(0) \neq 0$ due to $\Delta_{i,j,r}(x)$ having no roots at $0$. \footnote{We implore the reader to recall the motivating example in Section \ref{subsec:motivating_ex} as an instance of this problem.}
	
	By Definition \ref{def:task}, every $(i,j) \in \mathcal{S}$ is part of some $\mathcal{L}^q_A\times \mathcal{L}^q_B$. Thus, we are able to extract the matrices  $\{\mathbf{A}_i\mathbf{B}_j:(i,j) \in \mathcal{S}\}$ with the stated recovery threshold $R_{FCSA}$.

	Now, we discuss the system complexities. It is clear from the construction how we achieve the stated Upload Costs $(U_A,U_B)$, Download Cost $D_C$, and Worker Complexity $C_w$ by observing the sizes of the matrices $\widebf{A}_k$, $\widebf{B}_k$, and $\widebf{C}_k$. For the encoding complexity $C_A$, note that we only perform linear operations which can be summarized as performing $\alpha\beta$ matrix-vector operations with a $K \times L_A$ matrix which clearly achieves the desired encoding complexity $C_A$. Similar logic follows for $C_B$.  For the decoding complexity $C_d$, we first solve   $\alpha\gamma$  systems of  $R_{FCSA}$ linear equations defined by the interpolation problem in Eq.\eqref{eq:cauchy_van_function} which is known to have a complexity of $\mathcal{O}(\alpha\gamma R_{FCSA}\log^2(R_{FCSA})\log\log(R_{FCSA}))$ \cite{jia2021cross,YANG2002259,Olshevsky2000ASA,Gohberg1994FastAW}. The next stage in the decoding process involves solving $\alpha\gamma |\mathcal{Q}|$ system of $L^q_AL^q_B$ linear equations for every $q \in [|\mathcal{Q}|]$ defined by upper-triangular matrices in Eq. \eqref{eq:upper_triangular} which in total have a complexity of $\mathcal{O}(\alpha\gamma\sum_{q=1}^{|\mathcal{Q}|}(L^q_AL^q_B)^2)$. For decoding the triangular system of equations, we assume that the decoding is performed sequentially. Though, it is clear that we can speed-up the decoding by performing parallel inversions on the triangular systems. 
	Combining them all together results in the stated decoding complexity in Theorem \ref{theorem:fcsa_codes}. 
		\end{proof}
	
	\lev{
	\begin{remark}
		At a high-level, we can interpret Theorem \ref{theorem:fcsa_codes} as using the task assignment to delineate groups of computations that will be extracted, and that power assignment is used to align the terms to avoid interference with each other. Importantly, we note that the terms associated with the rational functions correspond to the desired computations in the task assignment and all the undesired computations are compacted within the polynomial terms. Interference alignment was first considered in \cite{jia2021cross}.  Unlike \cite{jia2021cross} where the problem space is static, we expand on this idea by incorporating power assignment which aims to minimize the number of polynomial terms to reduce their interference and, thus, allows for a smaller recovery threshold. Power assignment provides the necessary flexibility to account for the variable nature of VCDBMM.	\end{remark}
}
	\begin{remark}\label{remark:poly}
		A major motivating factor for constructing FCSA codes using task and power assignment is to be able to match or beat the recovery threshold of the codes defined in Section \ref{subsec:relevant_cons}. We will discuss important cases in Section \ref{section:fcsa_types} but we want to make a mention of one case for intellectual purposes. Consider the case when the task assignment only has one task grouping, i.e., the task grouping is the entire bipartite graph. Due to symmetry, it is not hard to determine that the best recovery threshold of FCSA codes for this case is $L_AL_B$ which is the same recovery threshold as for Polynomial codes defined in Section \ref{subsubsec:poly_codes}. Thus, FCSA codes can always achieve the recovery threshold of Polynomial codes. Though, in those cases, it is generally better to use Polynomial codes due to their faster decoding complexity by using fast polynomial interpolation. We mention this fact to show that FCSA codes are comparable to Polynomial codes in terms of the recovery threshold. In later sections, we will provided situations where FCSA codes are strictly better than the other codes discussed in Section \ref{subsec:relevant_cons}.
	\end{remark}
	\subsection{Power Assignment Optimization}
	We note that FCSA codes strongly depend on finding a good task  and  power assignments in a tractable manner, i.e., choosing $\mathcal{Q}$ and $\mathcal{P}$ to optimize the recovery threshold. Given a task assignment $\mathcal{Q}$, the optimization of the power assignment $\mathcal{P}$ to minimize the value of $R_{FCSA}$ can be straightforwardly designed as a binary linear program  \footnote{We note that this formulation as a BLP is an improvement over our conference paper \cite{tauz2021variable} where we had a binary quadratically constrained program with a linear objective. BLPs are generally much faster to optimize. }
	(BLP) by applying the conditions in Definition \ref{def:power}. The full optimization problem can be found in Appendix \ref{app:ilp}. While the formulation as a BLP does not guarantee a tractable method to find the optimal solution, we find that sub-optimal BLP solvers provide fairly good results \cite{Achterberg2013,mitchell2002branch,gurobi}. It is a possible future research direction to determine tractable methods to find the optimal power assignment. In subsequent sections, we will demonstrate some special cases of FCSA codes by providing simple methods to determine their task assignments. Despite being simple to construct, these special constructions provided significantly improved recovery thresholds in comparison to the previously discussed schemes. In the next section, we shall demonstrate how to generalize FCSA codes to allow flexible communication and computation complexity. 
	\section{FCSA with Flexible Communication and Computation}\label{sec:fcsa_flex}
	
	In this section, we demonstrate how we can use the construction of FCSA codes in Section \ref{sec:fcsa} to create a coding scheme with flexible communication and worker complexity which we term FCSA with Flexible Communication and Computation (FCSA+FCC). This new scheme is an extension of FCSA codes. The reason we provide the base construction of FCSA codes in Section \ref{sec:fcsa} before the more general version is to highlight the construction that provides the most gain over the other schemes in Section \ref{subsec:relevant_cons} for comparable communication and computational complexity, which will be demonstrated in later sections. Additionally, since we construct FCSA+FCC codes using FCSA codes as building blocks, preceding technical discussion carries over to this section and improves clarity.

	First, we present an important concept in the theory of fast matrix multiplication \cite{blaser2013FastMatrixMultiplication}.
	\begin{definition} \label{def:bilinear}
		(Bilinear Complexity \cite{blaser2013FastMatrixMultiplication}) Let $\mathbf{C} = (C_{j,k})^{k\in[n]}_{j\in[m]}$ be the matrix product of any matrices $\mathbf{A}= (A_{j,k})^{k\in[p]}_{j\in[m]}$ and $\mathbf{B}= (B_{j,k})^{k\in[n]}_{j\in[p]}$ where $C_{j,k} = \sum_{l=1}^{p}A_{j,l}B_{l,k}$. The \textit{bilinear complexity} of multiplying $\mathbf{A}$ and $\mathbf{B}$ is the minimum integer $R$, denoted as $R_{m,p,n}$, such that there exists tensors $a \in \mathbb{F}^{R\times m \times p}$,  $b \in \mathbb{F}^{R\times p \times n}$,  $c \in \mathbb{F}^{R\times m \times n}$ satisfying 
		\begin{align}
		\sum_{i=1}^{R}c_{i,j,k} & \underbrace{\left( \sum_{q=1}^{m}\sum_{r=1}^{p}a_{i,q,r}A_{q,r} \right)}_{\widetilde{A}_i}  \underbrace{\left(\sum_{r=1}^{p}\sum_{s=1}^{n}b_{i,r,s}B_{r,s}\right)}_{\widetilde{B}_i} \nonumber \\
		&= \sum_{l=1}^{p}A_{j,l}B_{l,k} = C_{j,k} \quad \forall j \in[m],k \in[n]. \label{eq:bilinear}
		\end{align}
	\end{definition}
	
	From this definition, we observe that bilinear complexity converts the problem of matrix multiplication into computing the element-wise product of two vectors of length $R_{m,p,n}$. Thus, by finding $\{\widetilde{A}_i\widetilde{B}_i\}^{R_{m,p,n}}_{i=1}$ we can recover the matrix product by applying the tensor $c_{i,j,k}$. While constructions of tensors known to achieve the optimal bilinear complexity do not yet exist for all $m,p,n$, there are many well known constructions that achieve an upper bound on the bilinear complexity such as Strassen's construction which provides an upper bound $R_{2^k,2^k,2^k} \leq 7^k$ \cite{strassen1969gaussian}. We remark that while the optimal $R_{m,p,n}$ is unknown for general $m,p,n$, $R_{m,p,n}$ is known to be sub-cubic in its parameters \cite{Yu2020EntangledPC}.
	
	Now, we can apply the concept of bilinear complexity to FCSA codes to achieve a trade-off between communication/computation complexity and recovery threshold. 
	
	Assume that we want to solve the VCDBMM problem with a computation list $\mathcal{S}$. Let $m,p,n \in \mathbb{Z}_{>}$ be fixed parameters that satisfy $m|\alpha$, $p|\beta$, and $n|\gamma$. For all $i \in [L_A]$ and $j\in [L_B]$, we partition the matrices in $\mathcal{A}$ and $\mathcal{B}$ as follows:
	\begin{equation}
	\mathbf{A}_i = \begin{bmatrix}
	\mathbf{A}_i^{1,1}& \cdots & \mathbf{A}_i^{1,p} \\
	\vdots& \ddots  & \vdots\\
	\mathbf{A}_i^{m,1} & \cdots & \mathbf{A}_i^{m,p}\\
	\end{bmatrix}
	, 
	\mathbf{B}_j =
	\begin{bmatrix}
	\mathbf{B}_j^{1,1}& \cdots & \mathbf{B}_j^{1,n} \\
	\vdots& \ddots  & \vdots\\
	\mathbf{B}_j^{p,1} & \cdots & \mathbf{B}_j^{p,n}\\
	\end{bmatrix}.
	\end{equation}
	
	Note that the partition $\mathbf{A}_i^{k,l}$ is of size $\frac{\alpha}{m} \times \frac{\beta}{p}$ for all $k \in [m],l\in [p]$ and $\mathbf{B}_j^{k,l}$ is of size $\frac{\beta}{p} \times \frac{\gamma}{n}$ for all $k \in [p],l\in [n]$. 
	
	As shown in Definition \ref{def:bilinear}, let $R = R_{m,p,n}$ be the bilinear complexity of a construction with tensors $a,b,c$ that satisfy Eg. \eqref{eq:bilinear}. Then, for all $ r \in [R], i\in [L_A], j\in [L_B]$,  we construct the encoding sub-matrices as 
	\begin{equation}
	\mathbf{A}^r_i =  \sum_{k=1}^{m}\sum_{l=1}^{p}a_{i,k,l}\mathbf{A}_i^{k,l}, \mathbf{B}^r_j =  \sum_{k=1}^{p}\sum_{l=1}^{n}b_{j,k,l}\mathbf{B}_j^{k,l}. \label{eq:bilinear_encoding}
	\end{equation}

	Thus, we can solve the VCDBMM problem by determining the matrix products $\{	\mathbf{A}_i^r\mathbf{B}_j^r: (i,j) \in \mathcal{S}, r \in [R]\}$ which is also a VCDBMM problem. An example of how to construct the new VCDBMM problem is shown in Fig. \ref{fig:bilinear_ex}. We observe that the new VCDBMM problem is essentially $R_{m,p,n}$ distinct VCDBMM problems. As such, by encoding the matrices using a construction for bilinear complexity, we can solve the resultant problem using FCSA codes. Afterward, we can decode the results using the values for $c_{i,j,k}$ and get the desired computations in $\mathcal{S}$. Thus, FCSA codes  can achieve flexible communication and computation complexity as shown in the following theorem.

	\begin{figure}
	\begin{center}
		\begin{tabular}{ c | c | c | c}
			\multicolumn{1}{c|}{}  &	$\mathbf{B}_1$ & $\mathbf{B}_2$ & $\mathbf{B}_3$ \\ \hline
 			$\mathbf{A}_1$ & \tikzmark{top left 1}1 & 1\tikzmark{bottom right 1} & 0  \\	\hline
			$\mathbf{A}_2$ & 0 & \tikzmark{top left 2}1 & 1\tikzmark{bottom right 2} \\  \hline
		\end{tabular} $\implies$
		\begin{tabular}{ c | c | c | c | c | c | c | c | c | c}
		\multicolumn{1}{c|}{}  &	$\mathbf{B}_1^1$ & $\mathbf{B}_2^1$ & $\mathbf{B}_3^1$  &	$\mathbf{B}_1^2$ & $\mathbf{B}_2^2$ & $\mathbf{B}_3^2$ &	$\mathbf{B}_1^3$ & $\mathbf{B}_2^3$ & $\mathbf{B}_3^3$\\ \hline
		$\mathbf{A}_1^1$ & 1 & 1 & 0 & 0 & 0 & 0 & 0 & 0 & 0\\	\hline
		$\mathbf{A}_2^1$ & 0 & 1 & 1 & 0 & 0 & 0& 0 & 0 & 0\\  \hline
		$\mathbf{A}_1^2$ & 0 & 0 & 0 & 1 & 1 & 0& 0 & 0 & 0\\  \hline
		$\mathbf{A}_2^2$ & 0 & 0 & 0 & 0 & 1 & 1& 0 & 0 & 0\\  \hline
		$\mathbf{A}_1^3$ & 0 & 0 & 0 & 0 & 0 & 0& 1 & 1 & 0 \\  \hline
		$\mathbf{A}_2^3$ & 0 & 0 & 0 & 0 & 0 & 0& 0 & 1 & 1\\  \hline
		\end{tabular}
	\end{center}
\caption{An example of using bilinear complexity to transform one VCDBMM problem into another VCDBMM problem. Assume that $R_{m,p,n}= 3$. \label{fig:bilinear_ex}}
	\end{figure}

		\begin{theorem}\label{theorem:fcsa_codes_flex} (Achievability of FCSA+FCC codes)
		For a given computation list $\mathcal{S}$, assume that a valid task assignment $\mathcal{Q}$ and associated power assignment $\mathcal{P}$ are provided for the original VCDBMM problem. Additionally, parameters $m,p,n \in \mathbb{Z}_{>}$ are provided such that $m|\alpha$, $p|\beta$, and $n|\gamma$. Let $R=R_{m,p,n}$ denote the bilinear complexity of multiplying an $m$-by-$p$ matrix and a $p$-by-$n$ matrix.  Furthermore, let $\rho$ be a parameter\footnote{We note that the case of $\rho =1$ was demonstrated in our conference paper \cite{tauz2021variable} .} such that $\rho|R$.  Then,  assuming that $|\mathbb{F}| > R_{m,p,n}|\mathcal{Q}| + K  $,  FCSA codes achieve the following:
		\hspace{-0.5em}
		\begin{align}
		&\text{Recovery Threshold: } \nonumber \\
		&R^{m,p,n}_{FCSA} = (R+\rho)\sum_{q=1}^{|\mathcal{Q}|}L^q_AL^q_B 
		- \min_{i\in [L_A]}\left(\sum_{q=1}^{|\mathcal{Q}|}P^{A,q}_{i}\right) 
		- \min_{j\in[L_B]}\left(\sum_{q=1}^{|\mathcal{Q}|}P^{B,q}_{j}\right) + 1 \\
		&\text{Upload Costs: }(U_A,U_B) =  \left(\frac{R}{\rho}\cdot \frac{\alpha\beta}{ mp},\frac{R}{\rho}\cdot \frac{\beta\gamma}{pn}\right)\\
		&\text{Download Cost: }  D_C = \frac{\alpha\gamma}{mn} \\
		&\text{Encoding Complexities: }  \begin{cases}
		C_A  &=  \mathcal{O}\left(\alpha\beta L_AR (\frac{K}{mp}+1)\right)\\
		C_B  &= \mathcal{O}\left(\beta\gamma L_BR(\frac{K}{np}+1)\right) \\
		\end{cases} \\
		&\text{Worker Complexity: }  C_w = \mathcal{O}\left(\frac{R}{\rho}\cdot \frac{ \alpha\beta\gamma}{mpn}\right) \\
		&\text{Decoding Complexity: }  C_d = \mathcal{O}(\alpha\gamma R|\mathcal{S}|)  + \widetilde{\mathcal{O}}(\frac{\alpha\gamma}{mn}(R^{m,p,n}_{FCSA}\log^2(R^{m,p,n}_{FCSA}) + R\sum_{q=1}^{|\mathcal{Q}|}(L^q_AL^q_B)^2)).
		\end{align}
	\end{theorem}
	
	\begin{proof}
		
		Assume that the input matrices are encoded as stated in Eq. \eqref{eq:bilinear_encoding}. Thus, we have $R_{m,p,n}$ individual VCDBMM problems. Let $\phi$ be defined such that $\rho\phi=R_{m,p,n}$. Partition the $R_{m,p,n}$ individual problems into $\phi$ groups with $\rho$ instances. For convenience, we can equivalently partition the set $[R_{m,p,n}]$ such that $\{1,2,\dots,\rho\}$ refer to the instances in the first partition, $\{\rho+1,\rho+2,\dots,2\rho\}$ refer to the instances in the second partition, and so forth. Additionally, we define $R=R_{m,p,n}$.
		
		Now, we encode each partition using the FCSA encoder. For clarity, we re-state the encoding process. First, let $f_1,f_2,\dots, f_{R|\mathcal{Q}|},x_1,x_2,\dots,x_K$ be distinct elements from $\mathbb{F}$. For all $i \in [L_A]$, $j\in [L_B]$, $h \in [\phi]$, $ l \in [\rho]$ we define 
		\begin{equation}
		a^{h,l}_i(x) = \prod_{q=1}^{|\mathcal{Q}|}(x-f_{((h-1)\rho+(l-1))|\mathcal{Q}|+q})^{P^{A,q}_{i}},
		b^{h,l}_j(x) =\prod_{q=1}^{|\mathcal{Q}|}(x-f_{((h-1)\rho+(l-1))|\mathcal{Q}|+q})^{P^{B,q}_{j}} .
		\end{equation}
		Additionally, for all $h \in[\phi]$ we define
		\begin{equation}
		\Theta_h(x) = \prod_{q=1}^{|\mathcal{Q}|}\prod_{l=1}^{\rho}(x-f_{((h-1)\rho+(l-1))|\mathcal{Q}|+q})^{L^q_AL^q_B} .
		\end{equation}
		For each partition $h \in [\phi]$, we define the encoding polynomials as 
		\begin{equation}
		\mathbf{A}^h(x) = \Theta_h(x) \sum_{i=1}^{L_A}\sum_{l=1}^{\rho}\frac{\mathbf{A}^{(h-1)\rho + l}_{i}}{a^{h,l}_i(x)} , \mathbf{B}^h(x) = \sum_{j=1}^{L_B}\sum_{l=1}^{\rho}\frac{\mathbf{B}^{(h-1)\rho + l}_{j}}{b^{h,l}_j(x)} . \label{eq:flex_encoding}
		\end{equation}

		For the $k^{\text{th}}$ worker, we send the following matrices:
		\begin{align}
		\widebf{A}_k = \begin{bmatrix}
		\mathbf{A}^1(x_k) & \mathbf{A}^2(x_k) & \cdots & \mathbf{A}^{\phi}(x_k)
		\end{bmatrix} , \widebf{B}_k = \begin{bmatrix}
		\mathbf{B}^1(x_k) \\ \mathbf{B}^2(x_k) \\ \vdots\\  \mathbf{B}^{\phi}(x_k)
		\end{bmatrix}
		\end{align} 
		where $\widebf{A}_k \in \mathbb{F}^{\frac{\alpha}{m}\times \phi\frac{\beta}{p}}$ and $\widebf{B}_k \in \mathbb{F}^{\phi\frac{\beta}{p}\times \frac{\gamma}{n}}$. Thus, the output $\widebf{C}_k  \in \mathbb{F}^{\frac{\alpha}{m}\times \frac{\gamma}{n}}$ of the $k^{\text{th}}$ worker is 
		\begin{equation}
		\widebf{C}_k  = \sum_{h = 1}^{\phi} \mathbf{A}^h(x) \mathbf{B}^h(x)
		\end{equation}

		From the proof of Theorem \ref{theorem:fcsa_codes}, we know that $ \mathbf{A}^h(x) \mathbf{B}^h(x)$ can be expressed as a summation of $\rho\sum_{q=1}^{|\mathcal{Q}|}L^q_AL^q_B $ rational terms and a polynomial of max degree $	\rho\sum_{q=1}^{|\mathcal{Q}|}L^q_AL^q_B  
		- \min_{i\in [L_A]}\left(\sum_{q=1}^{|\mathcal{Q}|}P^{A,q}_{i}\right) 
		- \min_{j\in[L_B]}\left(\sum_{q=1}^{|\mathcal{Q}|}P^{B,q}_{j}\right)$. Note that the rational terms all have unique roots in the denominator since each task assignment group within each VCDBMM instance is given a unique element $f_q$. Thus, $\sum_{h = 1}^{\phi} \mathbf{A}^h(x) \mathbf{B}^h(x)$ results in 
		\begin{equation}
		\sum_{h = 1}^{\phi}\rho\sum_{q=1}^{|\mathcal{Q}|}L^q_AL^q_B = \phi \rho\sum_{q=1}^{|\mathcal{Q}|}L^q_AL^q_B = R\sum_{q=1}^{|\mathcal{Q}|}L^q_AL^q_B
		\end{equation}
		unique rational terms and 
		\begin{equation}
		\rho\sum_{q=1}^{|\mathcal{Q}|}L^q_AL^q_B  
		- \min_{i\in [L_A]}\left(\sum_{q=1}^{|\mathcal{Q}|}P^{A,q}_{i}\right) 
		- \min_{j\in[L_B]}\left(\sum_{q=1}^{|\mathcal{Q}|}P^{B,q}_{j}\right)+1
		\end{equation}
		polynomial terms. By Lemma \ref{lemma:cauchy_vandermonde}, we can get the coefficients of the rational terms using $R^{m,p,n}_{FCSA}$ worker outputs. From this point, the proof continues analogously to the proof of Theorem \ref{theorem:fcsa_codes} to extract $\{	\mathbf{A}_i^r\mathbf{B}_j^r: (i,j) \in \mathcal{S}, r \in [R]\}$. Then, we use the relevant $c_{i,j,k}$ tensor for the bilinear complexity construction used in Eq. \eqref{eq:bilinear_encoding} to get $\{	\mathbf{A}_i\mathbf{B}_j: (i,j) \in \mathcal{S}\}$. Thus, we achieve the stated recovery threshold.
		
		We now discuss the other systems metrics. Again, we can easily calculate the Upload Costs $(U_A,U_B)$, Download Cost $D_C$, and Worker Complexity $C_w$ by observing the sizes of the matrices $\widebf{A}_k \in \mathbb{F}^{\frac{\alpha}{m}\times \phi\frac{\beta}{p}}$, $\widebf{B}_k \in \mathbb{F}^{\phi\frac{\beta}{p}\times \frac{\gamma}{n}}$, and $\widebf{C}_k  \in \mathbb{F}^{\frac{\alpha}{m}\times \frac{\gamma}{n}}$. Now, let use consider the encoding complexity of $C_A$. First, the complexity of applying the construction of bilinear complexity in Eq. \eqref{eq:bilinear_encoding} is $\mathcal{O}(\alpha\beta L_A R)$. Then, applying the linear encoding in Eq. \eqref{eq:flex_encoding} is $\mathcal{O}(\sum_{h = 1}^{\phi}\frac{\alpha\beta}{mp}L_A\rho K) = \mathcal{O}(\frac{\alpha\beta}{mp}L_A\phi\rho K) = \mathcal{O}(\frac{\alpha\beta}{mp}L_ARK)$. Combining these equations together, we get $C_A=  \mathcal{O}\left(\alpha\beta L_AR (\frac{K}{mp}+1)\right)$. Similar logic can be followed to get $C_B$. Finally, we consider the decoding complexity. Again, we can follow the steps in Theorem \ref{theorem:fcsa_codes} to understand that the decoding complexity of acquiring  $\{	\mathbf{A}_i^r\mathbf{B}_j^r: (i,j) \in \mathcal{S}, r \in [R]\}$ is $\widetilde{\mathcal{O}}(\frac{\alpha\gamma}{mn}(R^{m,p,n}_{FCSA}\log^2(R^{m,p,n}_{FCSA}) + R\sum_{q=1}^{|\mathcal{Q}|}(L^q_AL^q_B)^2))$. Afterwards, we apply the relevant $c_{i,j,k}$ tensor for the bilinear complexity construction used in Eq. \eqref{eq:bilinear_encoding} which has a complexity of $\mathcal{O}(\alpha\gamma R)$  for each matrix. Thus, the overall complexity of the final step is $\mathcal{O}(\alpha\gamma R|\mathcal{S}|)$. Combining these two terms together gets us the stated decoding complexity and completes the proof. 
	
		\lev{
		\begin{remark}
			The major novelty of this construction is that we can adjust the communication and computation costs using the parameters $m,p,n,$ and $\rho$. For example, the worker complexity is scaled by $\frac{R_{m,p,n}}{\rho m p n}$ and even if $\rho = 1$, $\frac{R_{m,p,n}}{\rho m p n}\leq 1$  due to $R_{m,p,n}$ being sub-cubic in its parameters. With the addition of $\rho$, we can control the worker complexity with a large array of terms between $\frac{1}{mpn}$ and $1$. Yet, this requires an appropriate increase in the recovery threshold to account for the reduction in communication and computation cost.  Additionally, we note that as $m,p,n$ increase, the dominant term in $R^{m,p,n}_{FCSA}$ becomes $(R_{m,p,n}+\rho)\sum_{q=1}^{|\mathcal{Q}|}L^q_AL^q_B $ which is independent of the power assignment. This indicates that the improvements offered by FCSA codes become less significant as the values for $m,p,n$ increase. Intuitively, this reduction happens because the bipartite graph of the new VCDBMM problem gets sparser for higher values of $R_{m,p,n}$, as can be seen in Fig. \ref{fig:bilinear_ex}. To allow for flexibility in the recovery threshold even for a high value of $R_{m,p,n}$, parameter $\rho$ can also be used to control the tradeoff between the varying metrics. For example, if $\rho= \frac{R_{m,p,n}}{2}$, then the communication and computation costs asymptotically go to $0$ while the dominant term in $R^{m,p,n}_{FCSA}$ becomes $(\frac{3R_{m,p,n}}{2})\sum_{q=1}^{|\mathcal{Q}|}L^q_AL^q_B $.
		\end{remark}
		}

	\end{proof}

	
	\section{Lower Bound on Optimal Recovery Threshold}\label{sec:lower_bound}
	
	To understand how well FCSA codes solve the VCDBMM problem, we shall provide a lower bound on the optimal recovery threshold. Our bound is stated for the general case of flexible communication and computation, i.e., when matrices are partitioned based on the values of $m,n,p$ which results in a specific download cost.
	\begin{theorem}\label{theorem:optimality}
		For a given $m,p,n,$ and $\mathcal{S}$, let $R^*_{\mathcal{S},m,p,n}$ be the optimal recovery threshold with a fixed download cost of $\frac{\alpha\gamma}{nm}$. Then, 
		\begin{equation}
		R^*_{\mathcal{S},m,p,n} \geq mn|\mathcal{S}|.
		\end{equation}
		
	\end{theorem}
	
	\begin{proof}
		See Appendix \ref{app:optimality_bound}. 
		The proof idea is showing the existence of a set of input matrices $\mathcal{A}$ and $\mathcal{B}$ that require $R^*_{\mathcal{S},m,p,n} \geq mn|\mathcal{S}|$. We accomplish this bound by fixing $\mathcal{B}$ to have full rank when its components are horizontally concatenated and let $\mathcal{A}$ be uniformly sampled from $\mathbb{F}$. We then apply a cut-set bound argument to show that the minimum number of worker outputs needed is $mn|\mathcal{S}|$ to guarantee enough symbols to recover the output.
	\end{proof}
	
	In the following section, we demonstrate special cases of FCSA codes and  use Theorem \ref{theorem:optimality} to show that these special cases achieve a recovery threshold within a multiplicative gap of $2$ for the case of $p=1$.  
	
	\section{Special Cases of FCSA codes} \label{section:fcsa_types}
	In this section, we consider special cases of FCSA codes by providing methods to construct task assignments. We analyze these cases and provide their relatively simpler expression for the recovery thresholds. Despite the simplicity of these constructions, they provide fairly significant improvements in terms of the recovery threshold in comparison to the relevant constructions in Section \ref{subsec:relevant_cons} as will be demonstrated in Section \ref{sec:sim}.
	
	Before describing these special cases of FCSA codes, we remind the reader of the graph theoretic notation defined in Section \ref{sec:prelim}. Specifically $d^A_i = |\{j: (i,j) \in \mathcal{S} \}|$ and $d^B_j = |\{i: (i,j) \in \mathcal{S} \}|$. Additionally, note that the task and power assignments only affect the recovery threshold and decoding complexity in Theorem \ref{theorem:fcsa_codes} and Theorem \ref{theorem:fcsa_codes_flex}. Thus, we will only mention these measures when discussing the special cases of FCSA codes.
	
	\subsection{Type-1 FCSA codes}
	The first special case is known as Type-1 FCSA (T1-FCSA) codes. These codes represent the worst-case upper bound on the recovery threshold. We now discuss how to construct T1-FCSA codes.  Let $|\mathcal{Q}| = |\mathcal{S}|$. Assume that we enumerate all the computations in $\mathcal{S}$. For $q \in [|\mathcal{S}|]$, let $\mathcal{L}^q_A = \{i_q\}$ and $\mathcal{L}^q_B = \{j_q\}$ which guarantees that the power assignments become $P^{A,q}_{i_q} = P^{B,q}_{j_q} = 1$. Essentially, we are assigning every computation in $\mathcal{S}$ to its own task grouping. An example of T1-FCSA codes is provided in Fig. \ref{fig:t1_ex}. By Theorem \ref{theorem:fcsa_codes_flex}, we have the following:
	\begin{theorem}
		(Type-1 FCSA+FCC codes)
		
		For T1-FCSA codes and parameters $m,p,n,$ and $\rho$, we have
		\begin{align}
		&R^{T1}_{FCSA} = (R_{m,p,n}+\rho)|\mathcal{S}|-\min_{i\in [L_A]: d^A_i\neq 0  }d^A_i - \min_{j\in[L_B]: d^B_i\neq 0 }d^B_j + 1,\label{eq:corollary_fcsa_t1_r} \\
		&C^{T1}_d =  \mathcal{O}(\alpha\gamma R|\mathcal{S}|)+ 
		\widetilde{\mathcal{O}}(\alpha\gamma R^{T1}_{FCSA}\log^2(R^{T1}_{FCSA})).
		\end{align}
	\end{theorem}
	\begin{proof}
		
		Since the number $P^{A,q}_{i}$ where $P^{A,q}_{i}  = 1$ is the number of $j \in [L_B]$ such that $(i,j) \in \mathcal{S}$, we must have that $\sum_{q=1}^{|\mathcal{Q}|}P^{A,q}_{i} =  |\{j: (i,j) \in \mathcal{S} \}| = d^A_i$. We ignore the terms where $d^A_i = 0$ since such a vertex would have been pruned from the graph. Similar logic follows for $P^{B,q}_{j}$. Thus, we get the stated recovery threshold. 
		
		Now, note that the change in decoding complexity comes from the fact that every grouping in the task assignment only contains one matrix product and, thus, there is no need to invert upper-triangular matrices.
	\end{proof}
	
	From this recovery threshold, we can state the following:
	\begin{corollary}\label{cor:t1_codes}
		When $p=1$, T1-FCSA are optimal within a multiplicative gap of 2. Specifically,
		\begin{equation*} 
			R^{T1}_{FCSA} \leq 2R^*_{\mathcal{S},m,1,n}.
		\end{equation*}
	\end{corollary}

\begin{proof}
	First, note that $\rho \leq R_{m,p,n}$. Thus, 
	\begin{equation}
	R^{T1}_{FCSA} = (R_{m,p,n}+\rho)|\mathcal{S}|-\min_{i\in [L_A]: d^A_i\neq 0 }d^A_i - \min_{j\in[L_B]: : d^B_j\neq 0 }d^B_j + 1 \leq 2 R_{m,p,n}|\mathcal{S}|.
	\end{equation}
	
	Due to Theorem \ref{theorem:optimality}, we only have to show that $R_{m,1,n} \leq mn$. This is straightforward since we can simply use the un-coded sub-matrices as an upper bound construction. Thus, $R_{m,1,n} \leq mn$ and the proof is complete. 
\end{proof}
	
	We also note that since T1-FCSA codes require no power assignment optimization, $R^{T1}_{FCSA}$ is always available as a worst-case upper-bound on the achievable recovery threshold among FCSA codes without any extra effort. In the next subsection, we demonstrate a coding scheme that uses power assignment optimization to find a better recovery threshold. Though we wish to remark that while T1-FCSA codes may not have the optimal recovery threshold among the FCSA codes, the reduced decoding complexity can provide an overall better completion time depending on $\mathcal{S},m,p,n,$ and $\rho$. 
	
	\begin{figure}
		\begin{subfigure}{0.5\textwidth}
			\centering
			\begin{tabular}{l |c || c | c | c |}
				\cline{3-5}
				\multicolumn{1}{c}{} & \multicolumn{1}{c||}{}  &	$\mathbf{B}_1$ & $\mathbf{B}_2$ & $\mathbf{B}_3$ \\ \cline{2-5}
				& \diagbox{$P^{A,q}$}{\vspace{-1em}\\ $P^{B,q}$}&	 
				\textcolor{group1}{1},\textcolor{group2}{0},\textcolor{group3}{0},\textcolor{group4}{0}&  
				\textcolor{group1}{0},\textcolor{group2}{1},\textcolor{group3}{1},\textcolor{group4}{0}& 
				\textcolor{group1}{0},\textcolor{group2}{0},\textcolor{group3}{0},\textcolor{group4}{1}\\\hline\hline
				\multicolumn{1}{|c|}{$\mathbf{A}_1$} &
				\textcolor{group1}{1},\textcolor{group2}{1},\textcolor{group3}{0},\textcolor{group4}{0}
				& \tikzmark{top left 1} 1 \tikzmark{bottom right 1} & \tikzmark{top left 2} 1 \tikzmark{bottom right 2} & 0  \\	\hline
				\multicolumn{1}{|c|}{$\mathbf{A}_2$} &
				\textcolor{group1}{0},\textcolor{group2}{0},\textcolor{group3}{1},\textcolor{group4}{1}
				& 0 & \tikzmark{top left 3} 1 \tikzmark{bottom right 3}  & \tikzmark{top left 4} 1 \tikzmark{bottom right 4} \\  \hline
			\end{tabular}
			\DrawBox[ultra thick, group1]{top left 1}{bottom right 1}
			\DrawBox[ultra thick, group2]{top left 2}{bottom right 2}
			\DrawBox[ultra thick, group3]{top left 3}{bottom right 3}
			\DrawBox[ultra thick, group4]{top left 4}{bottom right 4}
			\caption{Example of T1-FCSA Code.}\label{fig:t1_ex}
		\end{subfigure}
		\begin{subfigure}{0.5\textwidth}
			\centering
			\begin{tabular}{l |c || c | c | c |}
			\cline{3-5}
			\multicolumn{1}{c}{} & \multicolumn{1}{c||}{}  &	$\mathbf{B}_1$ & $\mathbf{B}_2$ & $\mathbf{B}_3$ \\ \cline{2-5}
			& \diagbox{$P^{A,q}$}{\vspace{-1em}\\ $P^{B,q}$}&	 
			\textcolor{group3}{2},\textcolor{group2}{0}&  
			\textcolor{group3}{1},\textcolor{group2}{1}& 
			\textcolor{group3}{0},\textcolor{group2}{2}\\\hline\hline
			\multicolumn{1}{|c|}{$\mathbf{A}_1$} & \textcolor{group3}{2},\textcolor{group2}{0} & \tikzmark{top left 1}1 & 1\tikzmark{bottom right 1} & 0  \\	\hline
			\multicolumn{1}{|c|}{$\mathbf{A}_2$} & \textcolor{group3}{0},\textcolor{group2}{2} & 0 & \tikzmark{top left 2}1 & 1\tikzmark{bottom right 2} \\  \hline
		\end{tabular}
		\DrawBox[ultra thick, group3]{top left 1}{bottom right 1}
		\DrawBox[ultra thick, group2]{top left 2}{bottom right 2}
			\caption{Example of T2-FCSA Code. }\label{fig:t2_ex}
		\end{subfigure}
	\caption{Examples of special cases of FCSA codes. The example of T2-FCSA code is also the motivating example in Section \ref{subsec:motivating_ex}.}
	\end{figure}

	\subsection{Type-2 FCSA codes}
	The second case improves on the recovery threshold of T1-FCSA codes by adding only a little complexity in optimizing the power assignment. We term this special case as Type-2 FCSA (T2-FCSA) codes. Let $|\mathcal{Q}| = L_A$. For $i \in [L_A]$, let $\mathcal{L}^i_A = \{i\}$ and $\mathcal{L}^i_B = \{j:(i,j) \in \mathcal{S}\}$ which guarantees that the power assignment has the property that $P^{A,q}_{i} = d^A_{i}$. Thus, only $P^{B,q}$ has to be optimized.  Note that the choice to partition based on $\mathcal{A}$ is arbitrary and the same task assignment can be done for $\mathcal{B}$. An example of T2-FCSA codes is provided in Fig. \ref{fig:t2_ex}. By Theorem \ref{theorem:fcsa_codes_flex}, we have the following:
	\begin{corollary}
		(Type-2 FCSA+FCC codes)
		
		For T2-FCSA codes and parameters $m,p,n,$ and $\rho$, we have
		\begin{align}
		&R^{T2}_{FCSA} = (R_{m,p,n}+\rho)|\mathcal{S}|-\min_{i\in [L_A]: d^A_i\neq 0 }d^A_i - \min_{j\in[L_B]}\left(\sum_{q=1}^{|\mathcal{Q}|}P^{B,q}_{j} \right)  + 1,\label{eq:corollary_fcsa_t2_r}\\
		&C^{T2}_d  =  \mathcal{O}(\alpha\gamma R|\mathcal{S}|)+
		\widetilde{\mathcal{O}}(\alpha\gamma(R^{T2}_{FCSA}\log^2(R^{T2}_{FCSA})+R_{m,p,n}|\mathcal{S}|^2)).
		\end{align}
	\end{corollary}
\begin{proof}
	The recovery threshold is straightforward to acquire by plugging in the values for $P^{A,q}_{i}$. The change in decoding complexity arises from the fact that for T2-FCSA codes the decoding complexity of inverting the triangular system of linear equations becomes $\sum_{q=1}^{|\mathcal{Q}|}(L^q_AL^q_B)^2 = \sum_{i=1}^{L_A}(d^A_i)^2 \leq |\mathcal{S}|^2$. 
\end{proof}

	It is straightforward to see that $R^{T1}_{FCSA} \geq R^{T2}_{FCSA}$ since $P^{B,q}_j \geq 1$ if $(i,j) \in \mathcal{S}$ which implies that $\min_{j\in[L_B]}\sum_{q=1}^{|\mathcal{Q}|}P^{B,q}_{j} \geq \min_{j\in[L_B]}d^B_j$. Thus, by Corollary \ref{cor:t1_codes}, we get the following:
	\begin{corollary}\label{cor:t2_codes}
		When $p=1$, T2-FCSA are optimal within a multiplicative gap of 2. Specifically,
		\begin{equation*} 
		R^{T2}_{FCSA} \leq 2R^*_{\mathcal{S},m,1,n}.
		\end{equation*}
	\end{corollary}

	We note that while these corollaries imply a bound on the achievable recovery thresholds of FCSA codes, these bounds are rarely tight. In the next section, we will provide numerical simulations to demonstrate the average multiplicative gap to optimality and display that FCSA codes perform much better in terms of the recovery threshold. 
	
	\begin{remark}
		As mentioned, Corollaries \ref{cor:t1_codes} and \ref{cor:t2_codes} are not necessarily tight for FCSA codes. This is especially true for FCSA+FCC when $\rho < R_{m,p,n}$. Consider the case when $\phi = \frac{R_{m,p,n}}{r} $ where $\phi > 1$. This case results in an increase in upload cost and worker complexity by a factor of $\phi$. Yet, the download cost is still the same and we can invoke Theorem \ref{theorem:optimality}. Hence, the multiplicative factor of optimality for $p=1$ can be upper bounded by $\frac{R^{T2}_{FCSA}}{R^*_{\mathcal{S},m,1,n}} \leq \frac{R_{m,1,n}+\rho}{R_{m,1,n}} =1 + \frac{1}{\phi}$. As such, any optimality gap can be satisfied for the appropriate increase in upload cost and worker complexity. Note that this upper bound does not take into account the usage of task and power assignment. As will be shown shortly, our new constructs can result in even further improvements in the recovery threshold. To avoid the issue of choosing values for $\rho$, we will concern ourselves with the case of $m=p=n=1$ and demonstrate how our code construction can improve upon the multiplicative optimality gap of 2.
	\end{remark}

	\section{Numerical Analysis of T1 and T2 FCSA codes}\label{sec:sim}

	In this section, we provide numerical analysis of the average multiplicative factor of optimality for T1-FCSA and T2-FCSA codes. The average is performed over the VCDBMM ensembles defined in Section \ref{sec:prelim}. We focus on the case when $m=p=n=1$ since that is the region where we can show the most benefit provided by FCSA codes. The metric that we will be focusing on is the ratio between the average recovery threshold of FCSA codes and the lower bound provided in Section \ref{sec:lower_bound}. Specifically, for the $V_{\lambda}(L_A,L_B)$ ensemble we analyze $G_{L_A,L_B,\lambda} = \frac{\mathbb{E}[R_{FCSA}]}{\mathbb{E}[|\mathcal{S}|]}= \frac{\mathbb{E}[R_{FCSA}]}{L_AL_B\lambda}$ and for the $V_{k}(L_A,L_B)$ ensemble we analyze $G_{L_A,L_B,k} = \frac{\mathbb{E}[R_{FCSA}]}{\mathbb{E}[|\mathcal{S}|]}= \frac{\mathbb{E}[R_{FCSA}]}{L_A\frac{1+k}{2}}$.
	
	For the $V_{\lambda}(L_A,L_B)$ and $V_{k}(L_A,L_B)$ ensembles, we shall focus on $\lambda$ and $k$ values that provide low to medium density within the bipartite graph which is the region of interest for our work. We note that in the complementary region that $2|\mathcal{S}|-1 \geq L_AL_B$ and, thus, FCSA codes are comparable to Polynomial codes as discussed in Remark \ref{remark:poly}. Additionally, we provide a baseline recovery threshold to compare against. We will be comparing to the LCC and CSA schemes discussed in Section \ref{subsec:relevant_cons}. We focus on these schemes since the values of $k$ and $\lambda$ we will analyze will highly likely satisfy $2|\mathcal{S}|-1 < L_AL_B$ and, thus, LCC and CSA codes will have a better recovery threshold than Polynomial codes. As such, for the $V_{\lambda}(L_A,L_B)$ ensemble we have $\frac{\mathbb{E}[R_{Baseline}]}{\mathbb{E}[|\mathcal{S}|]} = \frac{2\mathbb{E}[|\mathcal{S}|]-1}{\mathbb{E}[|\mathcal{S}|]} = 2 - \frac{1}{\mathbb{E}[|\mathcal{S}|]} = 2- \frac{1}{L_AL_B\lambda} $ and for the $V_{k}(L_A,L_B)$ ensemble we have $\frac{\mathbb{E}[R_{Baseline}]}{\mathbb{E}[|\mathcal{S}|]} = 2- \frac{1}{L_A\frac{1+k}{2}}$. Finally, when performing the simulations for T2-FCSA codes, we do power assignment optimization for both $\mathcal{A}$ and $\mathcal{B}$ partitioning and take the best recovery threshold. 	
	\begin{figure}[t]
		\centering
		\includegraphics[width=0.85\linewidth]{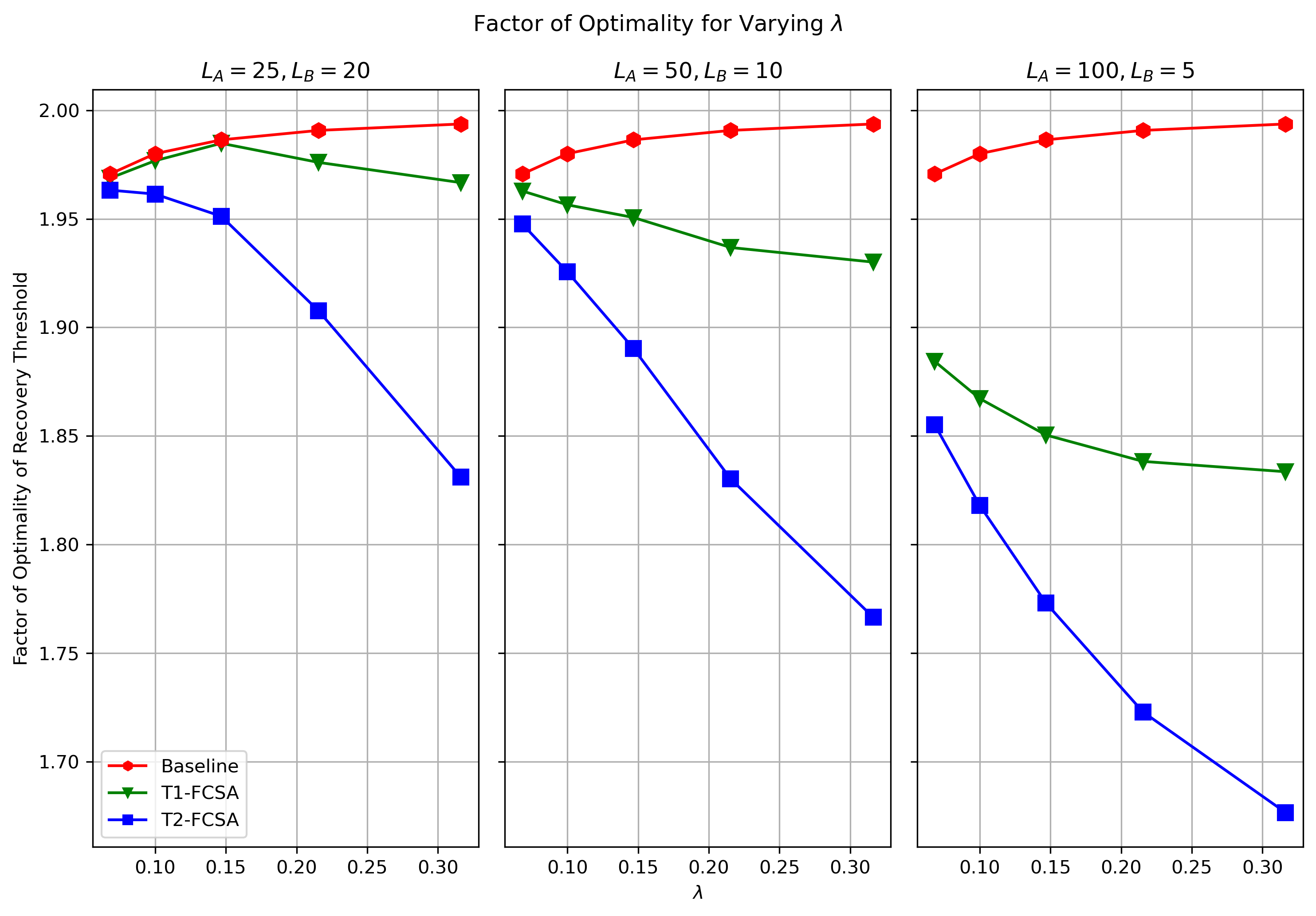}
		\vspace{-0.2cm}
		\caption{$G_{L_A,L_B,\lambda}$ for the $V_{\lambda}(L_A,L_B)$ ensemble with different $L_A$ and $L_B$ and varying values of $\lambda$.}
		\label{fig:random_plot_lamb}
		\vspace{-0.5cm}
	\end{figure}
	
	\begin{figure}[t]
		\centering
		\includegraphics[width=0.85\linewidth]{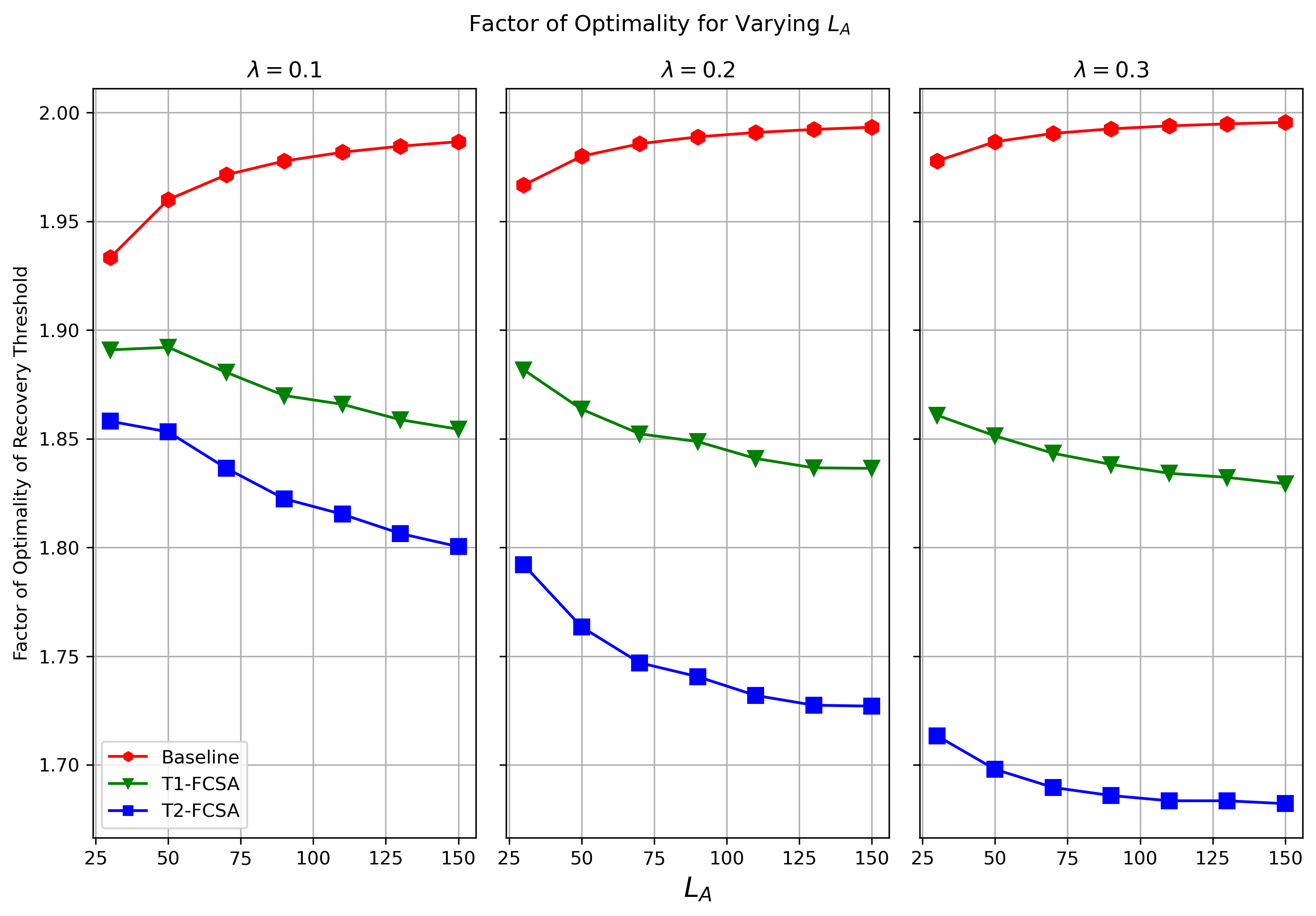}
		\vspace{-0.2cm}
		\caption{$G_{L_A,L_B,\lambda}$ for the $V_{\lambda}(L_A,L_B)$ ensemble with fixed $L_B=5$ and different $\lambda$ for varying values of $L_A$.}
		\label{fig:random_plot_la}
		\vspace{-0.5cm}
	\end{figure}

	\begin{figure}[t]
		\centering
		\includegraphics[width=0.85\linewidth]{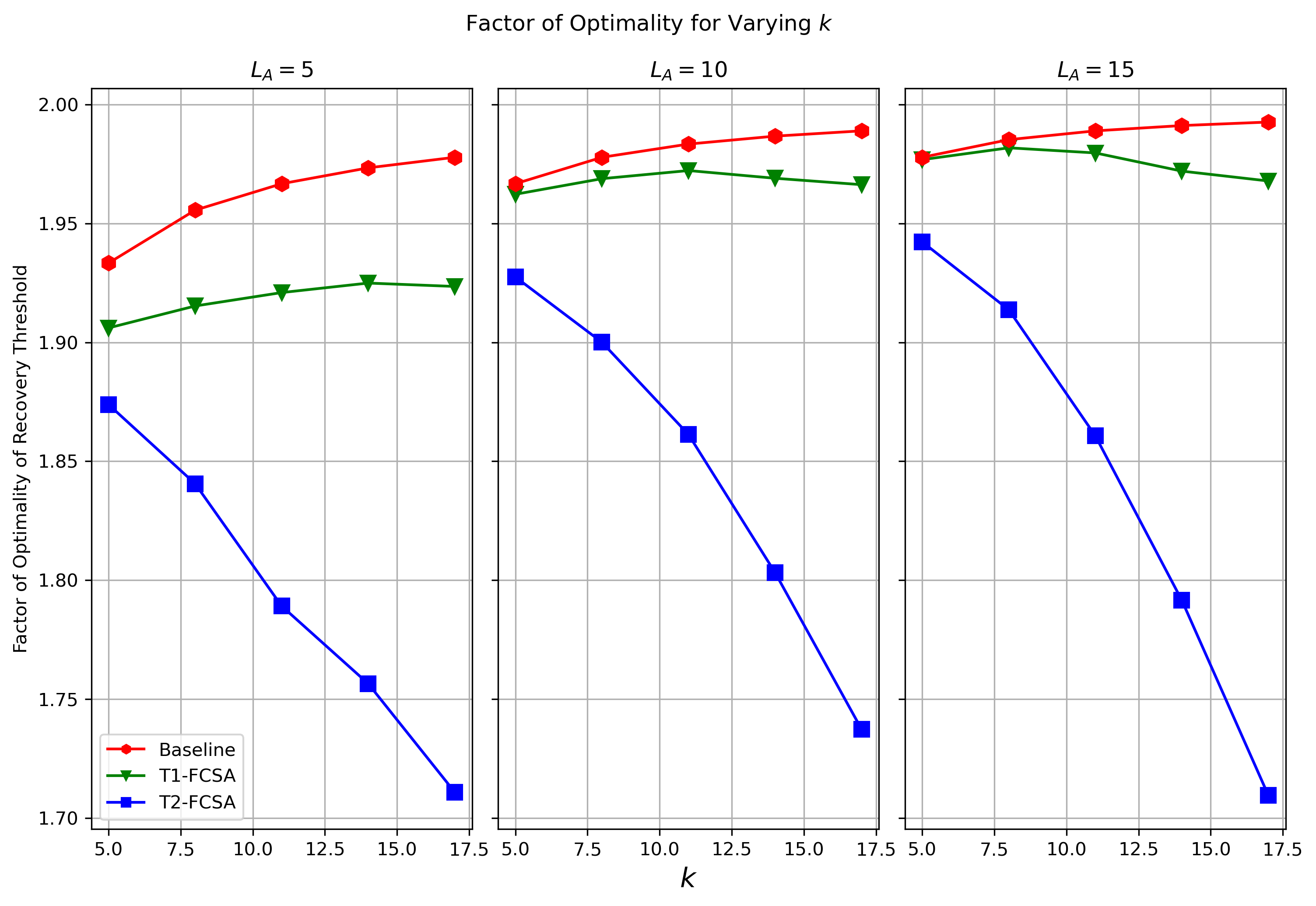}
		\vspace{-0.2cm}
		\caption{$G_{L_A,L_B,k}$ for the $V_k(L_A,L_B)$ ensemble with fixed $L_B=5$ and different $L_A$ for varying values of $k$.}
		\label{fig:random_plot_k}
		\vspace{-0.5cm}
	\end{figure}
	
	Fig. \ref{fig:random_plot_lamb} demonstrates the empirical simulation of $G_{L_A,L_B,\lambda}$ when $L_A$ and $L_B$ are fixed and we vary $\lambda$. The first observation we make is that T1-FCSA and T2-FCSA codes generally outperform the baseline recovery threshold. Additionally, T2-FCSA codes present significant improvement over T1-FCSA codes, with upto a difference of $0.15$ in terms of $G_{L_A,L_B,\lambda}$. This confirms that power assignment optimization can dramatically reduce the recovery threshold. Another observation is that $G_{L_A,L_B,\lambda}$ strongly depends on the values of $L_A$ and $L_B$ despite  $\mathbb{E}[|\mathcal{S}|]$ being exactly the same for all configurations. This suggests that there may exist a tighter bound than $|\mathcal{S}|$ that depends on the configuration of the computation edges. Finally, we observe that the curves for T2-FCSA codes follow almost a linear curve in terms of $\lambda$, indicating a possible scaling law to predict the average recovery threshold.
	
 	Fig. \ref{fig:random_plot_la} demonstrates the empirical simulation of $G_{L_A,L_B,\lambda}$ when $L_B=5$, $\lambda$ is fixed, and we vary $L_A$. This experiment was performed to further study the effect that the configuration of $L_A$ and $L_B$ have on the recovery threshold. Similar to the previous experiment, we see T1-FCSA and T2-FCSA codes outperform the baseline recovery threshold and T2-FCSA codes are significantly better than T1-FCSA codes. Interestingly, we see that this configuration creates a large gap between the baseline recovery threshold and the FCSA code recovery threshold. Thus, FCSA codes prove to be very useful when either $L_A$ or $L_B$ is kept small while the other can grow large. 
	
 	Fig. \ref{fig:random_plot_k} demonstrates the empirical simulation of $G_{L_A,L_B,k}$ when $L_A$ and $L_B$ are fixed and we vary $k$. Again, we see T1-FCSA and T2-FCSA codes outperform the baseline. Interestingly, we see that for T1-FCSA codes there is an increase in the value of $G_{L_A,L_B,k}$ as the density increases which is the opposite of the trend seen in the previous simulations. On the other hand, T2-FCSA codes improve rapidly as the value of $k$ increases. This further indicates a possible tighter bound that takes into account the structure of $V_k(L_A,L_B)$.
 	
 	From all these simulations, we can see that T1-FCSA and T2-FCSA on average have a multiplicative optimality gap that is generally much smaller than implied by Corollaries \ref{cor:t1_codes} and \ref{cor:t2_codes}. Additionally, both coding schemes outperform the relevant coding schemes discussed in Section \ref{subsec:relevant_cons} in terms of the recovery threshold.
	\section{Conclusion}
	In this paper, we presented the novel problem of Variable Coded Distributed Batch Matrix Multiplication and provided Flexible Cross-Subspace Alignment codes as novel solution with flexible parameters allowing for a variety of system complexities and near optimal straggler resilience. We provided a fundamental limit on the recovery threshold for VCDBMM and determined that for special cases the multiplicative optimaltiy gap for FCSA codes is a constant factor of $2$. Finally, we performed simulations on ensembles of the VCDBMM problem to demonstrate how FCSA codes generally provide much better recovery thresholds then our bounds imply and outperform other coding schemes that could be naively applied to solve VCDBMM.
	
	Despite the impressive performance of FCSA codes, there is still many avenues for possible future work. One major endeavor is on reducing the gap in optimality by determining a tighter lower bound as hypothesized in Section \ref{sec:sim}. Our current research direction is generalizing FCSA codes to other computation tasks besides matrix multiplication such as variable dot products. 
	\vspace{-0.5em}
	\section{Acknowledgments}
	Research supported in part by a grant from ASRC-IDEMA.

	
	\bibliographystyle{ieeetr}
	\bibliography{references_long}
	
	\appendix 
	
	\subsection{Proof of Theorem \ref{theorem:optimality}} \label{app:optimality_bound}

	We note that the problem definition specifies that the computation strategy has to work for all $\mathcal{A}$ and $\mathcal{B}$. We can thus prove the lower bound by constraining $\mathcal{A}$ and $\mathcal{B}$ to a certain class of matrices and providing a lower bound for this class which will in turn be a lower bound for the general case. Thus, assume that elements of $\mathcal{A}$ and $\mathcal{B}$ are from a finite field of size $u$, i.e., $\mathbb{F}_u$. Consider a matrix $\widehat{\mathbf{B}} \in  \mathbb{F}^{\beta\times L_B\gamma}_u$ that is constructed by horizontally concatenating all the matrices in $\mathcal{B}$, i.e., $\widehat{\mathbf{B}} = 
	\begin{bmatrix}
	\mathbf{B}_1 &  \cdots & \mathbf{B}_{L_B}
	\end{bmatrix}$. We assume that the matrices in $\mathcal{B}$ are chosen such that $\widehat{\mathbf{B}}$ is a tall matrix (i.e., $\beta > L_B\gamma$) and that $\widehat{\mathbf{B}}$ is full rank. Let $\mathcal{B}$ be fixed. Also, we define $C_{\mathcal{S}} = \{\mathbf{A}_i\mathbf{B}_j: (i,j) \in \mathcal{S}\}$.
	
	Let all the matrices in $\mathcal{A}$ be uniformly sampled from $\mathbb{F}^{\alpha\times\beta}_u$. As such, we can treat the symbols in $\mathbf{A}_i\mathbf{B}_j$ as random variables that are uniformly distributed on $\mathbb{F}^{\alpha\times\gamma}_u$. For a fixed $i \in [L_A]$, the computation $\mathbf{A}_i\mathbf{B}_j$ is a sub-matrix of $\mathbf{A}_i\widehat{\mathbf{B}}$. By the full rank property of $\widehat{\mathbf{B}}$, $\mathbf{A}_i\mathbf{B}_j$ are independent random variables for all $i\in [L_A]$ and $j \in [L_B]$. Therefore, the entropy of the desired computations is $H(C_{\mathcal{S}}) = |\mathcal{S}|\alpha\gamma \log_2{u}$. 
	We define $C_{\mathcal{R}} = \{\widebf{C}_i:i\in \mathcal{R}\}$ to be the set of outputs from the workers belonging to the index set $\mathcal{R} \subset [K]$. Assume that $\mathcal{R}$ is chosen such that the desired computations $C_{\mathcal{S}}$ can be decoded from $C_{\mathcal{R}}$ and that $|\mathcal{R}| =R^*_{\mathcal{S},m,p,n}$. This implies that $H(C_{\mathcal{S}}| C_{\mathcal{R}}) = 0$. Due to the fixed download cost of $\frac{\alpha\gamma}{mn}$, we know that the number of symbols in $C_{\mathcal{R}}$ is $R^*_{\mathcal{S},m,p,n}\frac{\alpha\gamma}{mn}$.
	We thus get the following:
	\begin{align}
	|\mathcal{S}|\alpha\gamma \log_2{u} = H(C_{\mathcal{S}}) = H(C_{\mathcal{S}}) - H(C_{\mathcal{S}}|C_{\mathcal{R}}) &= I(C_{\mathcal{S}};C_{\mathcal{R}}) \nonumber \\
	= H(C_{\mathcal{R}}) - H(C_{\mathcal{R}}|C_{\mathcal{S}})\leq   H(C_{\mathcal{R}})
	&\leq R^*_{\mathcal{S},m,p,n}\frac{\alpha\gamma}{mn} \log_2{u}
	\end{align}
	where $I(X;Y)$ is the mutual information and the last inequality comes from the sub-additivity of joint entropy and the fact that the uniform discrete distribution maximizes the entropy. As such, we get $mn|\mathcal{S}| \leq R^*_{\mathcal{S},m,p,n}$.

	\subsection{Power Assignment Optimization}  \label{app:ilp}
	
	Using Definition \ref{def:power} and Lemma \ref{lemma:power}, we note that the optimization problem of minimizing $R_{FCSA}$ can be written as the following optimization problem
	\begin{argmaxi!}|l|[3]
		{ P^{A,q}_i,P^{B,q}_j,a^q_i,b^q_j,t_q}{\min_{i\in [L_A]}\sum_{q=1}^{|\mathcal{Q}|}P^{A,q}_{i}
			+ \min_{j\in[L_B]}\sum_{q=1}^{|\mathcal{Q}|}P^{B,q}_{j}}{}{}
		\addConstraint{P^{A,q}_{i}=P^{B,q}_{j}}{= 0 \text{ if }  i \notin \mathcal{L}^q_A, j \notin \mathcal{L}^q_B, \quad \forall q \in [|\mathcal{Q}|] }
		\addConstraint{P^{A,q}_{i}}{= (t_q)(L^q_AL^q_B-a^q_i+1) + (1-t_q)(a^q_iL^q_B), \quad i\in \mathcal{L}^q_A ,q\in [|\mathcal{Q}|],}
		\addConstraint{P^{B,q}_{j}}{= (1-t_q)(L^q_AL^q_B-b^q_j+1) + (t_q)(b^q_jL^q_A), \quad j\in \mathcal{L}^q_B, q\in [|\mathcal{Q}|],}
		\addConstraint{P^{A,q}_{i} \neq P^{A,q}_{k} ,}{\forall i \neq k \in \mathcal{L}^q_A \quad \forall q\in[|\mathcal{Q}|]}
		\addConstraint{P^{B,q}_{j} \neq P^{B,q}_{l} ,}{\forall j \neq l \in \mathcal{L}^q_B \quad \forall q\in[|\mathcal{Q}|]}
		\addConstraint{a^q_i\in [L^q_A] , b^q_j }{\in [L^q_B]  \quad , (i,j) \in \mathcal{L}_A\times \mathcal{L}_B,q\in [|\mathcal{Q}|] ,}
		\addConstraint{t_{q}}{\in \{0,1\} \quad \forall q\in [|\mathcal{Q}|] ,}
	\end{argmaxi!}
	where $t_q$ is a binary indicator which indicates which condition is active in the second bullet point of Definition \ref{def:power} with $t_q=1$ indicating that $(i)$ is active.
	
	Using standard re-formulation techniques, the above problem can be transformed into the following optimization problem:

\begin{argmaxi!}|l|[3]
	{ P^{A,q}_i,P^{B,q}_j,a^q_i,b^q_j,t_q,v^q_{i,k},x^q_{j,k},y,z}{y+z}{}{}
	\addConstraint{y \leq \sum_{q=1}^{|\mathcal{Q}|}P^{A,q}_{i},}{\quad i \in [L_A]}{}
	\addConstraint{z \leq \sum_{q=1}^{|\mathcal{Q}|}P^{B,q}_{j},}{\quad j \in [L_B]}{}
	\addConstraint{P^{A,q}_{i}=P^{B,q}_{j}}{= 0 \text{ if }  i \notin \mathcal{L}^q_A, j \notin \mathcal{L}^q_B, q \in [|\mathcal{Q}|] }
	\addConstraint{P^{A,q}_{i}}{= (t_q)(L^q_AL^q_B-a^q_i+1) + (1-t_q)(a^q_iL^q_B), \quad i\in \mathcal{L}^q_A ,q\in [|\mathcal{Q}|], \label{eq:optidef_pa}}
	\addConstraint{P^{B,q}_{j}}{= (1-t_q)(L^q_AL^q_B-b^q_j+1) + (t_q)(b^q_jL^q_A), \quad j\in \mathcal{L}^q_B, q\in [|\mathcal{Q}|],\label{eq:optidef_pb}}
	\addConstraint{a^q_{i}= \sum_{k=1}^{L^q_A}k\cdot v^q_{i,k},}{\quad q\in [|\mathcal{Q}|],  i \in \mathcal{L}^q_A }{}
	\addConstraint{b^q_{j}= \sum_{k=1}^{L^q_B}k\cdot x^q_{j,k},}{ \quad  q\in [|\mathcal{Q}|], j \in \mathcal{L}^q_B }{}
	\addConstraint{1= \sum_{k=1}^{L^q_A}v^q_{i,k},}{\quad  q\in [|\mathcal{Q}|],  i \in \mathcal{L}^q_A  \label{eq:optidef_reform_uniq_1}}{}
	\addConstraint{1= \sum_{i \in \mathcal{L}^q_A}v^q_{i,k},}{\quad q\in [|\mathcal{Q}|], k \in [L^q_A] \label{eq:optidef_reform_uniq_2}}{}
	\addConstraint{1= \sum_{k=1}^{L^q_B}x^q_{j,k},}{\quad  q\in [|\mathcal{Q}|], j \in \mathcal{L}^q_B \label{eq:optidef_reform_uniq_3} }{}
	\addConstraint{1= \sum_{j \in \mathcal{L}^q_B}x^q_{j,k},}{\quad  q\in [|\mathcal{Q}|],k \in [L^q_B] \label{eq:optidef_reform_uniq_4}}{}
	\addConstraint{t_{q}}{\in \{0,1\} \quad q\in [|\mathcal{Q}|] ,}
	\addConstraint{v^q_{i,k}}{\in \{0,1\} \quad q\in [|\mathcal{Q}|] , i\in \mathcal{L}^q_A, k \in [L^q_A]}
	\addConstraint{x^q_{j,k}}{\in \{0,1\} \quad q\in [|\mathcal{Q}|] ,j\in \mathcal{L}^q_B, k \in [L^q_B]}
\end{argmaxi!}
	Note that the only free variables are $t_{q},v^q_{i,k},x^q_{j,l},y,z$ of which $t_{q},v^q_{i,k},x^q_{j,l}$ are binary and $y,z$ can be treated as real values. The way we got this formulation is using the permutation polytope to use binary variables to select values for $a^q_i$ and $b^q_j$. This allows for better branching since we can add a constraint to avoid branches where $a^q_i = a^q_k$ for $i\neq k$. 
	
	We note that this optimization problem is very close to a binary linear program except for Eqs. \eqref{eq:optidef_pa} and \eqref{eq:optidef_pb}. Fortunately, we can reformulate these equations further to create only linear constraints. Consider Eq. \eqref{eq:optidef_pa} and observe that $0 \leq P^{A,q}_i \leq L^q_AL^q_B$ for all possible configurations. We can thus  re-write the quadratic equality in Eq. \ref{eq:optidef_pa} with the following linear inequalities:
	\begin{align}
	P^{A,q}_i &\leq L^q_AL^q_B - a^q_i +1 +(1-t_q)(L^q_AL^q_B) \label{eq:lin_1}\\
	P^{A,q}_i &\leq a^q_iL^q_A+(t_q)(L^q_AL^q_B) \label{eq:lin_2}\\
	P^{A,q}_i &\geq L^q_AL^q_B - a^q_i +1 -(1-t_q)(L^q_AL^q_B) \label{eq:lin_3}\\
	P^{A,q}_i &\geq a^q_iL^q_A-(t_q)(L^q_AL^q_B) \label{eq:lin_4}
	\end{align}
	Note that when $t_q=1$ then Eqs. \eqref{eq:lin_2}\eqref{eq:lin_4} become superfluous constraints and Eqs. \eqref{eq:lin_1}\eqref{eq:lin_3} become 
	\begin{align}
	P^{A,q}_i &\leq L^q_AL^q_B - a^q_i +1 \\
	P^{A,q}_i &\geq L^q_AL^q_B - a^q_i +1 
	\end{align}
	which implies that $P^{A,q}_i =  L^q_AL^q_B - a^q_i +1$. A similar line of reasoning follows when $t_q=0$. We can also apply this for Eq. \eqref{eq:optidef_pb}. Thus, the power optimization can be written as a binary linear program for which there are many efficient solvers \cite{gurobi}.

	Other re-formulations are possible but we found this one to be preferable due to only binary free variables and that the permutation constraints Eqs. \eqref{eq:optidef_reform_uniq_1}\eqref{eq:optidef_reform_uniq_2}\eqref{eq:optidef_reform_uniq_3}\eqref{eq:optidef_reform_uniq_4} allow for a better branching method using Special Ordered Sets (SOS) thereby reducing optimization complexity.

\end{document}